%% file: UNTB-HDP.tex
\documentclass[10pt]{article}

\input{GlobalDefs}

\usepackage{amsmath, amssymb,amsthm,fullpage}
\usepackage{mathrsfs}
\usepackage{hyperref}

\usepackage{multirow}
\usepackage{graphicx}

\usepackage{cite}

\usepackage{color} 

\usepackage[labelfont=bf,labelsep=period,justification=raggedright]{caption}

\makeatletter
\renewcommand{\@biblabel}[1]{\quad#1.}
\makeatother

\date{}

\newcommand{\eqnref}[1]{Equation~\ref{eqn:#1}}
\newcommand{\eqnlabel}[1]{\label{eqn:#1}}

\newcommand{\figref}[1]{Figure~\ref{fig:#1}}

\newcommand{\tabnum}[1]{\ref{tab:#1}}
\newcommand{\tabref}[1]{Table~\tabnum{#1}}

\usepackage{array}

\begin{document}

\begin{flushleft}
{\Large
\textbf{Linking statistical and ecological theory: Hubbell's unified neutral theory of biodiversity as a hierarchical Dirichlet process}
}
\\
Keith Harris$^{1}$,
Todd L Parsons$^{2}$,
Umer Z Ijaz$^{3}$,
Leo Lahti$^{4}$,
Ian Holmes$^{5}$,
Christopher Quince$^{6,*}$
\\
\bf{1} School of Mathematics and Statistics, University of Sheffield, Sheffield, UK
\\
\bf{2} Laboratoire de Probabilit\'{e}s et Mod\`{e}les Al\'{e}atoires, CNRS UMR 7599, UPMC Univ Paris 06, Paris, France
\\
\bf{3} Infrastructure and Environment Research Division, School of Engineering, University of Glasgow, Glasgow, G12 8LT, UK
\\
\bf{4} Department of Veterinary Biosciences, University of Helsinki, Helsinki, Finland \& Laboratory of Microbiology, Wageningen University, Wageningen, Netherlands
\\
\bf{5} Department of Bioengineering, University of California, Berkeley, California, USA
\\
\bf{6} Warwick Medical School, University of Warwick, Coventry, CV4 7AL, UK
\\
$\ast$ E-mail: c.quince@warwick.ac.uk
\end{flushleft}

\section*{Abstract}

Neutral models which assume ecological equivalence between species provide null models for community assembly. 
In Hubbell's Unified Neutral Theory of Biodiversity (UNTB), 
many local communities are connected to a single metacommunity through differing immigration rates. 
Our ability to fit the full multi-site UNTB has hitherto been limited by the lack of a computationally tractable 
and accurate algorithm. We show that a large class of neutral models with this mainland-island structure 
but differing local community dynamics converge in the large population limit to the hierarchical Dirichlet 
process. Using this approximation we developed an efficient Bayesian 
fitting strategy for the multi-site UNTB. 
We can also use this approach to distinguish between neutral local community assembly given a non-neutral metacommunity distribution 
and the full UNTB where the metacommunity too assembles neutrally.
We applied this fitting strategy to both tropical trees and a data set comprising 
570,851 sequences from 278 human gut microbiomes. The tropical tree data set was consistent with the UNTB but for the human gut neutrality was rejected at the whole community level. 
However, when we applied the algorithm to gut microbial species within the same taxon at different levels of taxonomic resolution, we found that species abundances within some genera were almost consistent with local community assembly. This was not true at higher taxonomic ranks. 
This suggests that the gut microbiota is more strongly niche constrained than macroscopic organisms, with different groups adopting different functional roles, but within those groups diversity may at least partially be maintained by neutrality.
We also observed a negative correlation between body mass index and immigration rates within the family Ruminococcaceae.
This provides a novel interpretation of the impact of obesity on the human microbiome as a relative increase in the 
importance of local growth versus external immigration within this key group of carbohydrate degrading organisms.

\section*{Introduction}

A key question in ecology is what maintains species diversity in communities. The classical view is that every species 
occupies a distinct niche and the species observed in a community are then determined by the niches present. 
The niche itself is viewed as an $n$-dimensional hyper-volume in a space of abiotic and biotic environmental variables 
\cite{hutchinson57}. If two species 
occupy the same niche then one will outcompete the other \cite{hardin60}. This viewpoint has been challenged by neutral theory.
Neutral models of species abundance combine stochastic population dynamics with the assumption of ecological equivalence between 
species, formally defined as equivalent forms for all {\it per capita} demographic rates, e.g., birth and death.
Ecological equivalence is assumed to operate between species with a similar functional role deriving from the same broad functional 
 group or guild of species \cite{simberloff91}.
The result of the neutrality assumption is that rather than one species always outcompeting another the abundances within 
the neutral guild fluctuate.
The diversity at a single site is then generated as a balance between the immigration of new species and local extinction \cite{caswell76}.
 In Hubbell's Unified Neutral Theory of Biodiversity (UNTB) these ideas were extended to multiple sites \cite{hubbell01} 
using a mainland-island structure \cite{macarthur67}. The local communities experiencing neutral dynamics are coupled through migration 
to a metacommunity where neutral dynamics are again assumed but diversity is generated through speciation on a longer time-scale. 

The relative importance of niche versus neutral processes in macroscopic organisms is controversial. The first attempts to address this 
question fitted the UNTB to species abundance distributions (SADs) from a single site and compared model fit to non-neutral alternatives, 
e.g., log-normal or log-series \cite{mcgill03}. 
The development of Etienne's genealogical approach, which allowed the calculation of an exact sampling formula or likelihood for a single-site UNTB model \cite{etienne04},   
was key in allowing the UNTB to be fit efficiently to abundance data \cite{etienne04,etienne05}. Maximising this likelihood 
with respect to the model parameters generates a model fit. 
However, single samples do not provide enough information to reliably fit the UNTB \cite{rosindell12} and it has been demonstrated that niche models can generate identical SADs to a single-site neutral model \cite{chisholm10}. A more powerful test of the UNTB is to fit a data set from multiple sites simultaneously assuming the same metacommunity but different immigration rates. The genealogical approach has been generalised to multiple sites with identical migration rates \cite{etienne07} but for the fully general case of multiple sites with different immigration rates the resulting sampling formula is computationally intractable for more than a few sites \cite{etienne09}. Instead, an approximate two-stage method has to be used \cite{munoz07,jabot08,Etienne09a}. 

If the importance of neutrality is still an open question for macroscopic organisms then it is even more pertinent for microbes. It is only 
the recent coupling of molecular methods for characterising species identity with next generation sequencing 
that has allowed the efficient determination of microbial community structure {\it in situ} \cite{hamady08}. 
However, we are now regularly generating data sets comprising hundreds of sites and tens of thousands of sampled individuals per site \cite{turnbaugh09}. In order to accurately fit the  multi-site UNTB to these data we developed an alternative to the likelihood based genealogical approach. We are able to show that the UNTB is, in the limit of large population sizes, equivalent to a model from machine learning, the hierarchical Dirichlet process (HDP) \cite{Teh2006}. Moreover, our result is more general than the UNTB, as this limit applies irrespective of the exact local community dynamics, provided species are neutral and the total community size is fixed. We can use this result to adapt the existing Bayesian fitting strategy for the HDP to the problem of fitting the UNTB \cite{jabot08}.

Using this strategy it is possible to efficiently fit even the largest data sets in a reasonable amount of time with the added advantage of generating full posterior distributions over the parameters rather than just a maximum likelihood prediction. 
This method also reconstructs the metacommunity distribution enabling us to separate the key question of whether a community appears neutral into two parts. We can generate samples from the full neutral model with our fitted parameters and, as in \cite{etienne07}, compare their likelihood with that of the observed samples to test for neutrality, but we can also generate samples given the observed metacommunity and, hence, test for neutral local community assembly alone.

We will validate this method by applying it to twenty-nine tropical tree plots from Panama \cite{condit2002}. We will then use it to determine the extent to which gut microbial communities are neutrally assembled \cite{turnbaugh09}. The human gut is not a closed system, being constantly subjected to immigration events mainly through the diet, hence a metacommunity description is appropriate. However, it is not obvious for microbes at what level we would expect neutrality to operate, as different types of microorganisms perform very different roles. Indeed, there is evidence of clustering of gut microbiota into different enterotypes \cite{arumugam11,holmes12,ding14}, which implies non-neutral structuring at the whole 
community level. We will address this issue by subdividing the species according to their taxa at multiple taxonomic levels. There is increasing evidence of ecologial coherence at higher taxonomic levels for bacteria, with particular taxonomic groupings correlating with broad traits and metabolic functions \cite{fierer07,philippot09,philippot10}. Thus, even though within a species there may be variability in gene content and the precise niche occupied by strains, e.g. commensal and pathogenic {\it Escherichia coli} \cite{rasko08}, at higher levels an ecological signal is preserved \cite{fierer07}. We will test whether this signal leads to species within taxa being distributed neutrally in the human gut. 

This is the first time that the full multi-site neutral model has been fit to microbial community data. Earlier studies fitted the proportion of sites that a given species was observed in as a function of its abundance in the metacommunity \cite{sloan2006}. However, this approach models local neutral community assembly only, cannot allow for different immigration rates between sites and does not utilise the actual abundances of species, only their presence or absence. Similarly, although \cite{woodcock2006} showed that the bacterial taxa-abundance distributions in tree-holes scaled across sites in a way that was consistent with the neutral model, they were not fitting to the actual species abundances directly, but rather the shapes of those distributions in individual sites. Recently, an attempt was made to determine the degree of neutrality in human gut microbiota but again by fitting the single-site distribution only \cite{jeraldo2012}. By testing for neutrality at both the local and metacommunity level, and by resolving to different taxonomic groups, we will address the question of what is structuring the newly revealed microbial diversity of the human gut.

\section*{Methods}

\subsection*{Hubbell's Unified Neutral Theory of Biodiversity (UNTB)}

The UNTB separates the dynamics in the metacommunity from that in the local communities but both are neutral. Assume that there are 
$M$ local communities indexed $i = 1,\ldots,M$ each with a fixed number of $N_i$ individuals. Each iteration 
of the local community dynamics for site $i$ comprises two steps: choose an individual at random and remove it; with probability 
$m_i$ migration occurs and this individual is replaced by a randomly chosen member of the metacommunity or with probability $1 - m_i$ 
it is replaced by a randomly chosen member of local community $i$.
A generation in the model consists of replacing each 
individual on average once which will require $N_i$ iterations of these two steps. 
These dynamics will generate a stochastic Markov chain for the abundance of each species \cite{mckane04}, 
which given a sufficiently long time will 
converge to a stationary, or time-invariant, distribution. In the UNTB it is assumed that the local communities are at this 
stationary state which we will denote as a vector for each site $\bar \pi_i$, 
with elements $(\pi_{i,1},\ldots, \pi_{i,S})$ giving the probability of observing a particular species at site $i$. 
The two parameters $m_i$ and $N_i$ can be conveniently replaced by a single immigration rate $I_i = \frac{m_i}{1 - m_i}(N_i - 1)$ \cite{etienne05}. 
The parameter $I_i$ controls the coupling of the local community to the metacommunity. As $I_i \rightarrow \infty $, 
the local community stationary distribution will approach the metacommunity distribution and the number of species at that site will increase, while as $I_i \rightarrow 0$, the local community will become dominated by a single species.

In the metacommunity equivalent neutral dynamics operate but with new species generated through speciation with a probability $\nu$. 
This occurs on a longer time-scale than the local community dynamics so that the metacommunity can be assumed fixed relative to the local communities. 
Just as in the local communities where $I_i$ is preferred to $m_i$, it is more convenient to use the speciation rate (or fundamental biodiversity 
number) to parameterise the metacommunity distribution, $\theta = \frac{\nu}{1 - \nu}(N - 1)$ \cite{etienne05}, where $N$ is the fixed number of 
individuals in the metacommunity. The parameter $\theta$ can be viewed as the rate at which new individuals are appearing in the metacommunity as a result of speciation. 
As it increases, the total number of species in the metacommunity also increases and the species abundance distribution becomes increasingly skewed to rare individuals.
The final component of the UNTB is to realise that the observed data, the $M \times S$ frequency matrix $\mathbf{X}$ with elements $x_{ij}$ giving the number of times 
species $j$ is observed at site $i$, is a sample from the local community \cite{etienne05}. The simplest approach is to assume sampling with replacement so that the multinomial distribution 
describes the vector of observations at a given site:
\begin{equation}
\bar X_i \sim MN(J_i, \bar \pi_i),
\end{equation}
where $J_i = \sum_{j=1}^S x_{ij}$ is the sample size.

\subsection*{The HDP limit to neutral metacommunities}

In the SI Appendix we show that a wide class of neutral models including the UNTB converge in the large population limit to the same 
hierarchical Dirichlet process (HDP) approximation. This approximation captures the essential hypothesis of the UNTB -- 
namely neutrality, finite populations, and multiple panmictic geographically isolated populations linked by rare migration -- 
whilst being robust to the specific details of the local community dynamics.  Analogous to the relationship between Kingman's coalescent, 
Kimura's diffusion, and the Wright-Fisher model and its many generalisations (\textit{e.g.}, Cannings' models), we find that under suitable conditions on the higher moments of the individual reproductive output (namely, that when one considers the corresponding
 genealogical process, the coalescent, mergers of three or more ancestral lines happen with vanishingly small probability as the 
population size tends to infinity), it is sufficient to introduce local effective population sizes for each deme to accurately 
approximate many disparate models. 

For example, just as Hubbell's UNTB has population dynamics analogous to the Moran model of population genetics, we could equally well consider a ``Wright-Fisher'' neutral model, in which all individuals perish at the end of each time step, but each leaves behind a Poisson distributed number of offspring (conditioned on the total population size).  Whilst qualitatively different, this model retains the notion of neutrality: each individual is equally likely to be the parent of a randomly chosen individual in the next generation.  With an appropriate choice of time rescaling (see Example 2 in the SI), this model also gives rise to the HDP in the large population limit, much as both the Moran and Wright-Fisher models give rise to the same diffusive limits for appropriate choices of effective population size.  By contrast, if we consider the highly-skewed reproduction model in which the offspring of one randomly chosen individual replaces all other individuals, we do not obtain the HDP, even though we preserve the neutral hypothesis - as we discuss in the SI (Section 1.2), we require that the offspring distribution is not so fat-tailed that one individual is reasonably likely to be parent to a significant portion of the next generation.  In this latter case, there is still a well-defined limit, but it is poorly understood; in particular, there is no known analogue to the Antoniak equation (Equation 6) upon which our approach rests. 

It has been shown previously that for large local population sizes, and assuming a fixed 
finite-dimensional metacommunity distribution with $S$ species present then 
the local community distribution, $\bar \pi_i$, can be approximated by a Dirichlet distribution \cite{sloan2006,sloan2007}. 
The parameters of this Dirichlet distribution are proportional to the 
immigration rate multiplied by the metacommunity distribution:
\begin{equation}
\bar \pi_i | I_i,\bar{\beta} \sim Dir(I_i\bar \beta), 
\end{equation}
where $\bar \beta = (\beta_1,\ldots,\beta_S)$ is the relative frequency of each species in the metacommunity.
In the SI Appendix (see Section 1.4: Corollary 1), we generalise this to the case where as for the UNTB,  
there is a potentially infinite number of species that can be observed in the local community. Then 
the stationary distribution is a Dirichlet process (DP) \cite{ferguson73}:
\begin{equation}
\bar{\pi}_i | I_i,\bar{\beta} \sim \mbox{DP} (I_i,\bar{\beta}).  
\end{equation}
The DP can be viewed as an infinite dimensional generalisation of the Dirichlet. It generates an infinite set of samples from the base distribution, which in this case is the metacommunity $\bar{\beta}$, while the concentration parameter, which is $I_i$ here, controls the distribution of weights of those samples. Indeed, these weights are generated by a stick-breaking process (see below) with parameter $I_i$.

In the metacommunity, a Dirichlet process also applies (SI Appendix: Section 1.5), but now the base distribution is simply a 
uniform distribution over arbitrary species labels, and the concentration parameter is the biodiversity parameter, $\theta$. 
This is not a new observation, as it is implicit in the use of 
Ewens's sampling formula \cite{ewens72} for the metacommunity in Etienne's approach \cite{etienne05}. In this case the 
metacommunity distribution is purely the stick-breaking process. Define an infinite set 
of random variables drawn from a beta distribution $\{\beta'_{k}\}_{k=1}^{\infty}$:
\begin{equation}
\beta'_{k} \sim Beta(1,\theta).
\end{equation}
Then we can define the $k^\mathrm{th}$ element of the metacommunity vector as:
\begin{equation}
\beta_k = \beta'_k.\prod_{l=1}^{k-1}(1 - \beta'_{l}).
\end{equation}
We will denote this process $\bar \beta \sim Stick(\theta)$. Since the local communities are also DPs the model becomes a hierarchical Dirichlet process (HDP) in the parlance of machine learning \cite{Teh2006}. The stick-breaking process is one way to view the DP but an alternative perspective can be obtained by considering successive draws from a DP, which yields the Chinese restaurant process, where each new draw has a probability proportional to the number of individuals already assigned to an existing type (which in our case would be species) of deriving from that type and a probability proportional to $\theta$ of deriving from a previously unseen type (or species). From this process the Antoniak equation for the number of types or species $S$ observed following $N$ draws from a DP with concentration parameter $\theta$ can be derived: 
\begin{equation}
P(S | \theta, N) = s(N,S) \theta^S \frac{\Gamma(\theta)}{\Gamma(\theta + N)}
\eqnlabel{Antoniak}
\end{equation}
where $s(N,S)$ is the unsigned Stirling number of the first kind \cite{Antoniak1974} and $\Gamma(x)$ denotes the gamma function.

\subsection*{Gibbs sampler for the Neutral-HDP model}

Combining the model elements described above, we obtain the complete Neutral-HDP model as:
\begin{align*}
\bar{\beta} | \theta &\sim \mbox{Stick} (\theta),\\
\bar{\pi}_i | I_i,\bar{\beta} &\sim \mbox{DP} (I_i,\bar{\beta}),\\
\bar{X}_i | \bar{\pi}_i, J_i & \sim \mbox{MN} (J_i,\bar{\pi}_i).
\end{align*}
To this we add gamma hyper-priors for the biodiversity parameter, $\theta$, and the immigration rates, $I_i$: 
\begin{align}
\theta | \alpha,\zeta &\sim \mbox{Gamma} (\alpha,\zeta), \eqnlabel{theta_prior} \\
I_i | \eta, \kappa\  &\sim \mbox{Gamma} (\eta,\kappa),
\end{align}
where $\alpha,\zeta, \eta \text{ and } \kappa$ are all constants. 

In any given sample although the potential number of species is infinite we only observe $S$ different types. It is convenient therefore to represent the model in terms of these finite dimensional number of types and one further class corresponding to all unobserved species. We will represent the proportions of the $S$ observed species explicitly as $\beta_k$ with $k = 1,\ldots,S$ and the unrepresented component as $\beta_u = \sum_{k = S + 1}^L \beta_k$, in the limit as $L \rightarrow \infty$.
In this finite dimensional representation we can determine the species distributions in the local communities:
\begin{equation}
\bar{\pi}_i ~ \sim \mbox{Dir}(I_i\beta_1,\ldots,I_i\beta_S,I_i\beta_u).
\end{equation}
We can then marginalise the local community distributions and derive the probability of the observed frequencies given the metacommunity distribution $\bar{\beta}$ and the immigration rates $I_i$, $i = 1, \dots, M$:
\begin{equation}
P(\mathbf{X}|\bar{\beta},I_1,\ldots,I_M) = \prod_{i = 1}^M \frac{J_i!}{X_{i1}! \cdots X_{iS}!}\frac{\Gamma(I_i)}{\Gamma(J_i + I_i)} \prod_{j = 1}^S \frac{\Gamma(x_{ij} + I_i\beta_j)}{\Gamma(I_i\beta_j)}.
\eqnlabel{marginallikelihood}
\end{equation} 

The observation that the UNTB is actually a hierarchical Dirichlet process allows us to utilise an efficient Gibbs sampling method to fit it. 
A Gibbs sampler is a type of Bayesian Markov chain Monte Carlo (MCMC) algorithm. An MCMC algorithm generates samples from the posterior distribution of the parameters given the data \cite{mackay92}, which in this case is $P(\theta, I_1,\ldots,I_M| \mathbf{X})$. In general, the posterior is too complex to sample from directly and, in Gibbs sampling, 
samples are instead generated from the conditional distribution of one parameter given all the others. These full conditionals are often much simpler than the joint posterior distribution, and, crucially, if repeated samples are taken in this way, then they will converge onto the posterior after sufficient iterations. By introducing extra auxiliary variables, it is possible to devise an efficient Gibbs sampler for the UNTB-HDP approximation. One of these auxiliary variables is the metacommunity distribution itself $\bar \beta$ and the other is the number of ancestors in site $i$ that gave rise to species $j$, denoted $T_{ij}$, i.e., the number of independent immigration events from the metacommunity. Using these variables a Gibbs sampling iteration proceeds as follows:
\begin{enumerate}

\item Sample the biodiversity parameter $\theta$ from the conditional:
\begin{equation}
P(\theta | S, T) \propto s(T,S) \theta^S \frac{\Gamma(\theta)}{\Gamma(\theta + T)} \mbox{Gamma} (\theta|\alpha,\zeta),
\eqnlabel{eqn:species_dist}
\end{equation}
where $T = \sum_{i=1}^M \sum_{j=1}^S T_{ij}$. 
The first part of the above expression derives from the Antoniak equation (\eqnref{Antoniak}) for the number of unique species observed, $S$, when we sample $T$ ancestors from the metacommunity Dirichlet process with concentration parameter, $\theta$, the second part is simply the prior on $\theta$ \cite{Antoniak1974}. To sample from this we use the auxiliary variable approach of~\cite{EscobarWest95}.

\item Sample the metacommunity distribution:
\begin{equation}
\bar{\beta} = (\beta_1, \beta_2, \dots, \beta_S, \beta_u) \sim \mbox{Dir} (T_{\cdot 1}, T_{\cdot 2}, \dots, T_{\cdot S}, \theta),
\eqnlabel{DirichletDistBeta}
\end{equation}
where $T_{\cdot j} = \sum_{i=1}^M T_{ij}$. 
This exploits the conjugacy between the stick breaking prior for the metacommunity, $\bar{\beta}$, and the likelihood of the ancestor numbers $T_{ij}$ \cite{Teh2006}.  

\item Sample the immigration rates:
\begin{equation}
P(I_i | T_{ij}) \propto \frac{\Gamma(I_i)}{\Gamma(J_i + I_i)}I_i^{T_{i \cdot}}\mbox{Gamma}(I_i|\eta,\nu).
\end{equation}
This is again just Antoniak's equation multiplied by the prior but here the number of unique types observed, are the ancestors from the metacommunity, $T_{i \cdot} = \sum_{j=1}^S T_{ij}$, in $J_i$ samples from the local community DP with concentration parameter, $I_i$. 

\item Sample the ancestral states:
\begin{equation}
\eqnlabel{eqn:Tables2}
P(T_{ij} | x_{ij}, I_i, \beta_j) = \frac{\Gamma(I_i\beta_j)}{\Gamma(x_{ij} + I_i\beta_j)} s(x_{ij}, T_{ij}) (I_i \beta_j)^{T_{ij}},
\end{equation}
where again we recognise the Antoniak equation. This summarises the Gibbs sampling but in SI Appendix 2 we rigorously derive the above conditional distributions.
\end{enumerate}
In general we found that this MCMC procedure quickly converges but to ensure that we were sampling from the stationary distribution we generated either 50,000 Gibbs samples for each fitted data set and discarded the first 25,000 iterations as burn-in or for the human gut microbiota when testing multiple taxa we used 10,000 Gibbs sample and discarded 5,000 iterations as burn-in. The results below are quoted as the median values over these last 25,000 or 5,000 samples with upper and lower  credible (Bayesian confidence) limits given by the 2.5\% and 97.5\% quantiles of these samples. 

An MCMC approach was used in an early method to fit the single-site model \cite{etienne04}, but it required the use of the more complicated Metropolis-Hastings algorithm, not Gibbs sampling, which is central to the efficiency of our method. In SI Appendix Section 2 we present detailed results demonstrating that on samples generated from the UNTB with known parameters that our method outperforms the two-stage approximate method of \cite{Etienne09a}, providing accurate and reliable estimates of both $\theta$ and $I_i$ except when $I_{i} \gg \theta$. In this case there is a consistent bias towards under-estimating $I_i$, which, as we explain in SI Appendix Section 2, is preferable to the large variation in the parameter estimates exhibited by the two-stage approximation. The HDP method also has two further advantages: it generates a full posterior distribution of the model parameters, which provides a realistic estimate of the uncertainty around their point estimates, and it also recovers the metacommunity distribution. 

To determine whether an observed data set appears neutral we used a similar Monte Carlo significance test to that in \cite{etienne07}. Given the $k^{\mathrm{th}}$ posterior sample of fitted UNTB parameters, $\theta^k, I_1^k,\ldots,I_M^k$, an artificial data matrix with the same number of samples $M$ and the same sample sizes $J_i$ as the original data matrix is generated by sampling from the full neutral-HDP, which we will denote by $\mathbf{X}^{k}_0$. Given this sample we can also generate a neutral metacommunity distribution, ${\bar \beta}^{k}_0$, using \eqnref{DirichletDistBeta}, since the ancestral frequencies $T_{\cdot j} = \sum_{i=1}^M T_{ij}$ are known. This will be a true neutral metacommunity since the distribution will correspond to stick-breaking with parameter $\theta$. Note that the number of species observed can differ from $S$. We then calculate the likelihood $P(\mathbf{X}^{k}_0|{\bar \beta}^{k}_0,I_1^k,\ldots,I_M^k)$ using \eqnref{marginallikelihood}. These likelihoods were then compared to the actual likelihood of the observed sample, $P(\mathbf{X}|{\bar \beta}^{k},I_1^k,\ldots,I_M^k)$, and the proportion that were smaller than that value calculated to give a pseudo p-value, denoted $p_N$, to test the null hypothesis of neutrality, such that a small $p_N$ indicates that the data is not consistent with the neutral model. In addition, we generated data sets, $\mathbf{X}^{k}_1$, with the metacommunity fixed at the model fitted values, ${\bar \beta}^k$. Due to the hierarchical nature of the model, the metacommunity DP only gives a prior on the metacommunity distributions, the observed meta-community can deviate from the neutral expectation. This enables us to test for local neutral community assembly but with a fitted potentially non-neutral metacommunity. We do this in the same way calculating the likelihood for each of the samples, $P(\mathbf{X}^{k}_1|{\bar \beta}^{k},I_1^k,\ldots,I_M^k)$, and comparing to $P(\mathbf{X}|{\bar \beta}^{k},I_1^k,\ldots,I_M^k)$, the proportion of samples with likelihood smaller than this forms our pseudo p-value for local neutral community assembly, which we denoteby $p_L$. 
For both tests, samples were generated either from 2,500 sets of fitted parameters taken from every tenth iteration of the last 25,000 Gibbs samples or from 500 sets of fitted parameters taken from every tenth iteration of the last 5,000 Gibbs samples for the human gut microbiota when testing multiple taxa.

There are many ways in which a distribution could appear non-neutral. A clear example is provided by the situation where communities fall into a finite number of distinct types such that community configurations cluster together. It has been suggested that the human gut microbiome can be clustered into three distinct enterotypes \cite{arumugam11,holmes12,ding14}. This will appear non-neutral since a single metapopulation distribution will be unable to desribe all the community configurations observed. In addition, communities can also appear non-neutral at the level of the observed taxa abundances, if the abundances within individual samples are more or less skewed to rare species than expected for a Dirichlet process then this will appear non-neutral at the local community level. If this occurs for the metacommunity then neutrality will be rejected there too.

\subsection*{Identifying neutral subsets of species}

For the microbial community data, we will separate species by their taxa and fit the model to taxa separately in an attempt to identify neutral subsets. The validity of this 
approach rests on two observations. Firstly, that if there are multiple neutral guilds of species in a community, where the abundance of a guild varies from site to site in a 
non-neutral fashion, then the community as a whole will appear non-neutral but if we just sample species from one guild then the neutral patterns will be recovered \cite{walker07}. 
This is self-evident. The second observation is that if only a subset of the species in a neutral guild are sampled, then that subset will still fluctuate neutrally but 
with renormalised probabilities. This derives from the following property of the Dirichlet distribution, that if only a subset of the $S$ dimensions are observed, say $U$, 
then that subset is still distributed as a Dirichlet on the reduced space with the same parameters. For the neutral model the result is that the biodiversity parameter 
is unchanged but that the immigration rate at each site is reduced, $I_i^U = I_i(1 - \sum_{i \notin  U}\beta_i)$, according to the weight of the missing species in the metacommunity. The result is that if at 
some level of taxonomic resolution all species are from the same neutral guild, if not necessarily representing all that guild, then they will still be identified as neutral.

The key ideas used in the above derivations are summarised in \tabref{Ideas}.
\subsection*{Data}

\subsubsection*{Neutral simulation} 

In SI Appendix Section 2 we show that the UNTB-HDP fitting method accurately determines the parameters of data sets generated from the UNTB. To provide a further test of the model fitting 
from a sample that relaxes the mainland-island structure of the UNTB but maintains the assumption of neutrality we performed a neutral model simulation. This comprised 
50 sites indexed $i = 1,\ldots, 50,$ with a fixed population number of $N_i = 20,000$ individuals per site. Discrete dynamics were used with a probability that an individual was removed at each iteration of 5\%. 
Deleted individuals were then replaced, with speciation probability $\nu = 10^{-5}$ by an entirely new species, by an individual chosen at random from the local community in the previous iteration with probability $(1 - \nu)( 1- m_i)$, or by an individual chosen at random from all the other sites with probability $(1 - \nu)m_i$. The migration probability was varied across sites according to the rule $m_i = i\times10^{-4}$, so that the immigration rate, $I_i = m_i N_i = 2i$, varied from 2 to 100. The model was run for 2,000 generations, i.e., 40,000 iterations, at which point the 
species number appeared stationary, then 1,000 individuals were sampled with replacement from each site. The UNTB-HDP model was fit by Gibbs sampling to this data set as was the two-stage approximate method of \cite{Etienne09a}. 
This simulation although it has strictly neutral dynamics does not correspond exactly to Hubbell's UNTB because rather than an explicit mainland-island structure with diversity only generated in the metapopulation, it has speciation occuring in the local populations themselves, with a metapopulation which is an implicit aggregate of the local populations rather than an explicit distribution.

\subsubsection*{Tropical trees from Panama}

To provide a well-distributed sample of tropical trees at a regional level we took twenty-nine of the one hectare forest plots considered in \cite{condit2002}. These comprised all the one hectare samples from the Panama region with an elevation of less than 200 metres. This restriction ensured that all samples were from the same environment of lowland tropical forest. We also did not use data from the three larger Panama plots in order to maintain an even sampling at the regional level. Within each plot all trees $\ge$ 10cm in diameter were censused and their morpho-species recorded. The network of sample sites was spread across a $15 \times 50$ km region along the Panama canal, see \cite{pyke01} for details. A total of 13,263 trees were sampled from 367 species. The number of individuals observed in each plot ranged from a minimum of 302 to 647 with a median of 450. The UNTB-HDP model was fit to this data as described above. 

\subsubsection*{Human gut microbiota}

To compare with the tropical tree analysis we also fitted the UNTB-HDP model to a study of the gut microbiomes of twins and their mothers \cite{turnbaugh09}. 
These comprised fecal samples from 154 different individuals characterised by family and body mass index (BMI). Each individual 
was sampled at two time points approximately two months apart. The V2 hypervariable region of the 16S rRNA gene was amplified by PCR and then sequenced 
using 454. We reprocessed this data set filtering the reads, denoising and removing chimeras using the AmpliconNoise pipeline \cite{quince09,quince11}. 
This gave a total of 570,851 reads split over 278 samples, 
since out of the 308 collected samples thirty failed to possess any reads following filtering. The size of individual samples varied from just 53 to 10,580 with a median of 1,598. 
The number of unique sequences remaining following noise removal was 19,647. These were then taxonomically classified using the RDP stand-alone classifier of \cite{wang07}.
We constructed Operational Taxonomic Units (OTUs) at 3\% sequence difference using average linkage clustering to approximate species \cite{youssef09}. This was done for the entire data set generating 7,238 OTUs. We fitted the UNTB-HDP model to this data set.

To explore the impact of sample size and number on the ability of our pseudo p-values to correctly identify a community as non-neutral at the local and metacommunity levels we generated a series of subsampled data sets from this study. First, we selected at random without replacement either 20, 50, 100 or 200 samples from all those that had 1,000 reads or greater (247 in total). Then we generated a series of data sets where we sampled increasing numbers of individuals or reads from these selected samples, from 20 individuals per sample to 400 inclusive in increments of 20. We used sampling with replacement i.e. multinomial sampling so that expected OTU proportions were equal to those in the observed communities. For each number of samples and number of reads we generated ten replicate communities. We then fitted the UNTB-HDP model to these communities and tested for neutrality at the local and metapopulation level.

Starting with the full data set, we split the unique sequences according to the phylum to which they were classified, using a cut-off of 70\% bootstrap confidence. OTUs were then reconstructed 
at 3\% for each phylum and the UNTB-HDP fit to each phylum separately. We repeated this process for family and genus too. Only samples that had more than 150 representatives from a taxa were included in the analysis and the model was only fit to taxa that had at least 50 samples satisfying this criterion. This ensured a sufficiently large data set for parameters to be inferred and if a taxa dominates a neutral guild occupying a particular role we would expect it to appear in a large proportion of samples. We also generated ten replicate data sets from the full data set with the same number of samples and same number of reads per sample as the data sets split by taxa at each level. Applying the UNTB-HDP to these then gives us an equivalent bench-mark for the effect of subsampling on our ability to detect non-neutrality. We also did this for the tropical tree data. 

\section*{Results}

\subsection*{Neutral simulation}

In \figref{Simulation}, we give the immigration rates estimated by the UNTB-HDP fitting algorithm for the neutral simulation. From this single sample we are able to accurately predict the immigration rates across all the sites. The uncertainty in our predictions increases for higher $I_i$ but there is no consistent bias. In contrast, the two-stage approximation substantially underestimates the immigration rate as $I_i$ increases. This is most likely because although the simulation appears locally neutral ($p_L = 0.57$) as we would expect, the hypothesis that the neutral model applies at the metacommunity level too is rejected, $p_N = 0.0096$. The deviation from the mainland-island structure and the occurrence of speciation within the islands themselves results in a metacommunity distribution that deviates from the neutral stick-breaking process. This illustrates that in contrast to the two-stage approximation the UNTB-HDP model can still correctly 
predict immigration rates when neutral community assembly operates only at the local community level.

\subsection*{Tropical Trees from Panama}

By fitting the UNTB-HDP model to the twenty-nine tropical tree communities we found that they have a distribution of abundances across sites that is consistent with the neutral model at both the metacommunity and local community levels, $p_N = 0.81$ and $p_L = 0.23$. The median fitted $\theta$ obtained was 109.3. The median fitted immigration rates varied across sites from 20.69 to 76.93 with a median of 41.7. In \figref{Trees}, we use non-metric multi-dimensional scaling (NMDS) to position each community in two-dimensions in such a way as to preserve Bray-Curtis distances between communities. This was done using the metaMDS function of the vegan package in R \cite{R}. The fitted metacommunity distribution is also shown in this plot. The sites are represented as bubbles with size proportional to their fitted immigration rates and contours calculated using the ordisurf function. From this it is apparent that the communities with higher $I_i$ are in general more similar to the metacommunity. The fitted immigration rates are also related to the spatial location of the sites. Although there is no spatial location associated with the metacommunity, if we assign it to the location of the site with the highest $I_i$, site 14, and calculate the distance from this site to each of the others, then we find a significant negative correlation (p = 0.03) between distance and immigration rate.

\subsection*{Human gut microbiota}

In contrast to the tropical trees, the human gut samples do not appear neutral at the whole community level, $p_N = 0$ and $p_L = 0$. This was not purely an effect of the tropical trees comprising a data set of fewer samples and fewer individuals. Reducing the gut data set to an equivalent number of samples (29) with the same sizes we would still always reject neutrality at the metacommunity level, at the local level we observed a median $p_L$ of $0.062$ across the ten replicates. We would falsely fail to reject neutrality therefore but not as strongly as for the real tree data ($p_L = 0.23$). Therefore, we can conclude that the human gut is convincingly less neutral than tropical trees even accounting for the different sample numbers and sizes.

In \figref{Power} we show the impact of sample number and sample size on the pseudo p-values for the test of neutrality for whole community and local community assembly. With sufficient samples (i.e. at least 200) we have power to reject neutrality at both levels provided the sample size exceeds 150 but as sample number decreases our power to correctly reject neutrality particularly for local community assembly decreases. 

The results of subdividing the OTUs at different taxonomic levels and fitting the UNTB-HDP model are given in a nested format in \tabref{Twins}. The families associated with each phylum are indented below as are the genera in each family. 
We see some evidence that as we move down the taxonomic hierarchy from phyla, through families to genera, the subdivided communities 
appear more consistent with neutral local community assembly. We would reject local neutrality 
for both major phyla found in the human gut, the Bacteroides and Firmicutes, but there are two families out of four for which we cannot 
confidently reject neutral local community assembly at the 1\% level, the Bacteroidaceae and Incertae Sedis XIV, with $p_L = 0.03$ and $0.05$, respectively. 
At the level of genera, two out of three appear close to neutral at the local level, the exception being the Faecalibacterium.
This is not the case when we do not use the fitted metacommunities and instead test for both neutral local community assembly and a neutral metacommunity. Then for all data sets we would completely reject neutrality. The figures in parantheses give pseudo p-values for the equivalent complete data set randomly sampled down to the same size as the taxa. This gives us a benchmark to verify that these affects are not purely due to small sample sizes. From these we see that in all cases the probability of incorrectly concluding that the subsampled data set is neutral is less than 1\%.

To quantify how the metacommunity deviates from the neutral assumption for those data sets that appear locally neutral 
we compared the fitted metacommunities averaged over 500 Gibbs samples with the metacommunity observed in samples 
from the full neutral model with the equivalent parameters. These two distributions are shown in \figref{Metapop} for 
the three genera, Bacteroides, Blautia and Faecalibacterium. 
These distributions are shown as rank-abundance plots with the OTUs ordered in terms of the relative frequency 
with that frequency given on the y-axis, which is log-scaled. It is clear that the fitted metacommunities from the three genera 
all have a small number of highly abundant OTUs and then a long tail of rare OTUs. The neutral model cannot fit a metacommunity of 
this shape.

We also looked for correlations between the fitted immigration rates for the different taxa and the body mass index of subjects. 
No significant relationships were found at the genus level but for the family Ruminococcaceae a significant negative relationship 
was observed (p-value = 0.014 see \figref{BMII}). The same negative correlation was also observed for their
parent phylum the Firmicutes but it was slightly stronger (p-value = 0.007). 

\section*{Discussion}

The results clearly demonstrate the usefulness of the UNTB-HDP Gibbs sampler, its ability to fit large multi-sample data sets, and 
its robustness to deviations of the metacommunity from neutrality and the ability to detect those deviations whilst still correctly 
inferring immigration rates. The resulting significance tests and fitted parameters reveal a great deal 
about the ecology of the human gut microbiota in comparison to macroscopic organisms such as the tropical trees. 
The human gut is clearly much more strongly structured by functional niches. Only at the genus level 
do we see some evidence of neutral local community assembly in the gut, whilst tropical trees were well described 
by the neutral model without any subdivision of species. In some ways, 
this is to be expected, given the multiplicity of metabolic roles performed by the human microbiota we would not expect ecological 
equivalence at the whole community level. 
However, the borderline neutral patterns we did observe suggest the possibility that neutral local community 
assembly may be operating within the species occupying those roles, and that neutral processes may be responsible for 
maintaining some of the vast diversity that is observed in the human gut. This has to be a tentative conclusion as 
pattern does not imply process \cite{rosindell12}, but, regardless, the fact the observed abundances are 
consistent with the neutral model means that its importance for explaining fine-scale gut microbial diversity cannot be ruled out. 

It is important to address the question of whether the tests have the power necessary to detect non-neutrality. 
It is clear from \figref{Power} that as the number of samples in particular decreases it becomes hard to detect non-neutral distributions --- this is actually a strong motivation for the use of the UNTB-HDP which can be efficiently fit in the multi-site case. 
However, our benchmarking against the full gut data set allows us to conclude that some genera and the tropical trees appear more neutral than the equivalent sized complete gut microbiome. It is also important to note that the model was unable to detect the spatial signature in the tropical tree data as a deviation from neutrality. In the absence of that spatial information we would have included that a spatially inhomogenous metapopulation was sufficient to explain these patterns. That certainly motivates inference strategies for spatially explicit neutral models\cite{rosindell08}. 

It is highly significant that the metacommunity distributions could not be explained by the neutral process for any taxa. 
Instead, the metacommunity was dominated by a small number of very abundant OTUs, with in all cases the most abundant OTU possessing a 
relative abundance exceeding 10\% of the metacommunity. This may be a signature of non-neutral processes. The dominant OTUs may have a 
competitive advantage, or interactions with bacteriophages \cite{minot11} or the host immune system may be structuring these 
distributions \cite{quince13}, and that is skewing their apparent metacommunity abundance, 
or it may genuinely reflect the abundance of these organisms in the metacommunity perhaps coupled with an improved dispersal 
ability over their competitors.   

The parameters of the fitted models, in particular, the immigration rates, are also highly informative. For the Panamanian tree data set we 
showed that these correlated with spatial location of the sites. A strong effect of distance on community similarity was found in 
the original study  and a spatially explicit version of the neutral model was fit to the data \cite{condit2002}, but we have shown 
that even in the UNTB where space is only implicit, this signal can be recovered from the fitted immigration rates. For the gut microbiota 
samples, we have no spatial position, but here, remarkably, the immigration rates for the family Ruminococcaceae and phylum 
Firmicutes correlated negatively with body mass index. This provides an unique interpretation of the impact of obesity on the 
human gut microbiota: an increase in the rate of input of nutrients to the gut effectively results in an increase in microbial growth rates in 
the key carbohydrate metabolising group the Ruminococcaceae \cite{ze12} and these equate to a decrease in immigration rate relative to local birth. 

It is also instructive to compare immigration rates between fitted models. There has been debate as to the importance of dispersal 
on microbial community structure, the theory that ``everything is everywhere, but the environment selects'' \cite{finlay04}. However, 
comparing the tropical tree fits with the gut microbiota at the phylum level we find that the predicted immigration rates are comparable, 
implying that dispersal limitation may be just as important between human guts as it is between tropical forests. Interesting patterns 
also appear comparing immigration rates between gut taxa. They are much lower, for example, for the Bacteroides than the Firmicutes, 
probably reflecting the much higher tendency for the latter to be spore-forming.

Finally, whilst these results are of great interest in themselves, perhaps our most significant achievement is 
formally linking a model from ecology, the Unified Neutral Theory of Biodiversity, with a model from machine learning, 
the hierarchical Dirichlet process. In addition, by showing that the details of the local community dynamics are irrelevant for the 
HDP approximation to hold, provided the neutrality assumption is met, we may explain why we were able to fit communities as 
different as tropical trees and the gut microbiota. This strongly motivates the HDP as an ecological null model. What is more the 
mathematical structure of the HDP is easily extendable to for example, niche-neutral models or further hierachical levels.  Therefore, 
we believe that the connection we have made here will lead to an explosion of hierarchical Bayesian modelling in community ecology.

Software for fitting the UNTB-HDP can be downloaded from:\\
 \noindent https://github.com/microbiome/NMGS.

\section*{Acknowledgments}

CQ is funded through an EPSRC Career Acceleration Fellowship EP/H003851/1. KH was funded through a Unilever research grant whilst conducting this research. LL is funded by the Academy of Finland (decision 256950). TLP is funded by the Fondation Sciences Math\'ematiques de Paris. We thank three anonymous reviewers for constructive comments.



\section*{Figures}

\begin{figure}[ht]
\begin{center}
	\includegraphics[width=3in]{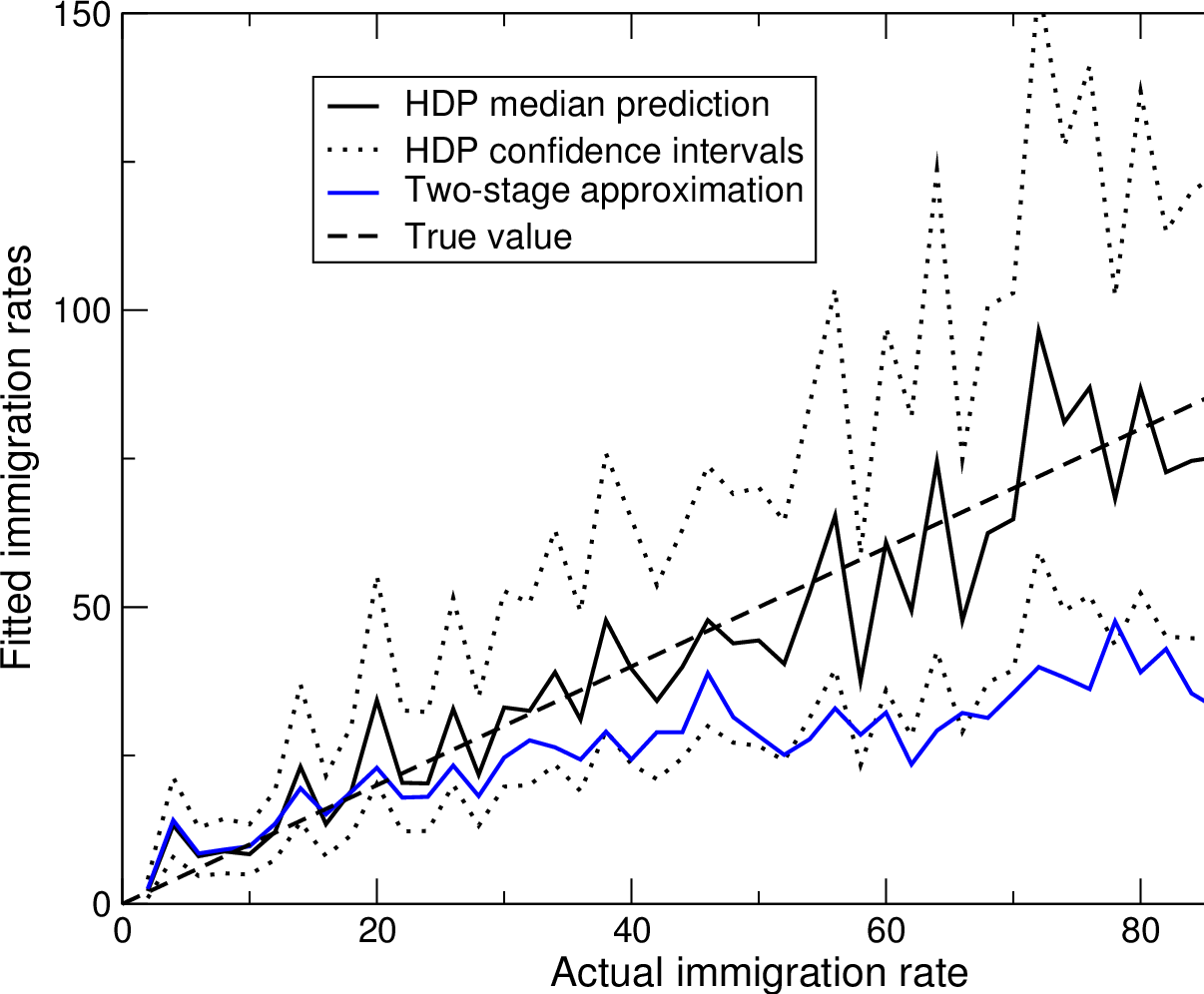}
\end{center}
\caption{{\bf Estimated immigration rates vs. true values for the UNTB-HDP model fit to a neutral model simulation.} Predictions are medians (solid line) from 25,000 posterior samples together with lower (2.5\%) and upper (97.5\%) Bayesian confidence intervals (dotted lines). The predictions from the two-stage approximation are also given (blue line).}
\label{fig:Simulation}
\end{figure}

\begin{figure}[ht]
\begin{center}
	\includegraphics[width=4in]{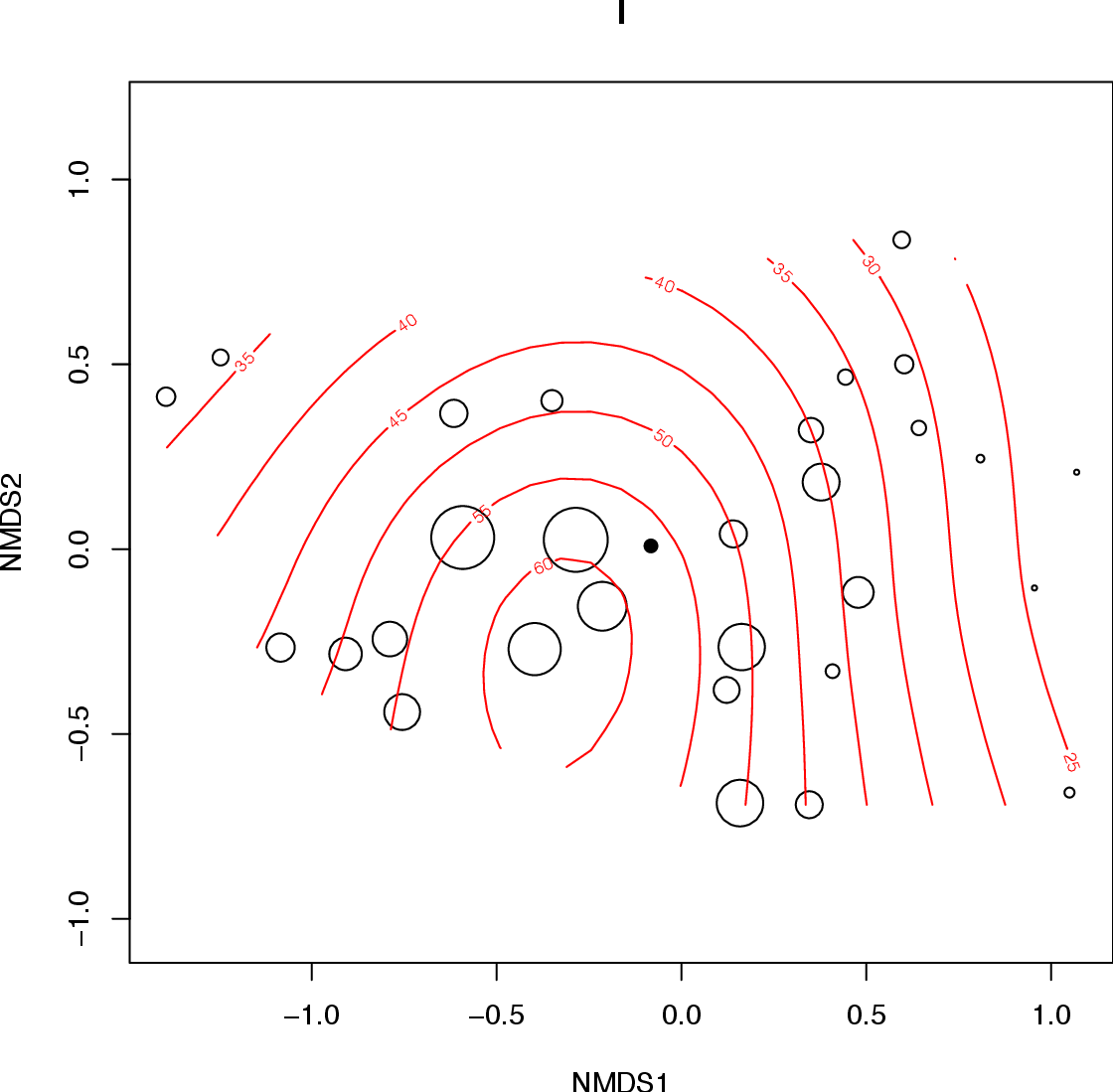}
\end{center}
\caption{{\bf An NMDS plot of the twenty-nine Panama tropical tree communities.} Communities are visualised as bubbles with size proportional to the median $I_i$ values obtained from the UNTB-HDP Gibbs sampler. Contours calculated using the ordisurf function of the R vegan package are also shown. The metacommunity distribution is denoted by a solid black point.}
\label{fig:Trees}
\end{figure}

\begin{figure}[ht]
\begin{center}
	\includegraphics[width=4in]{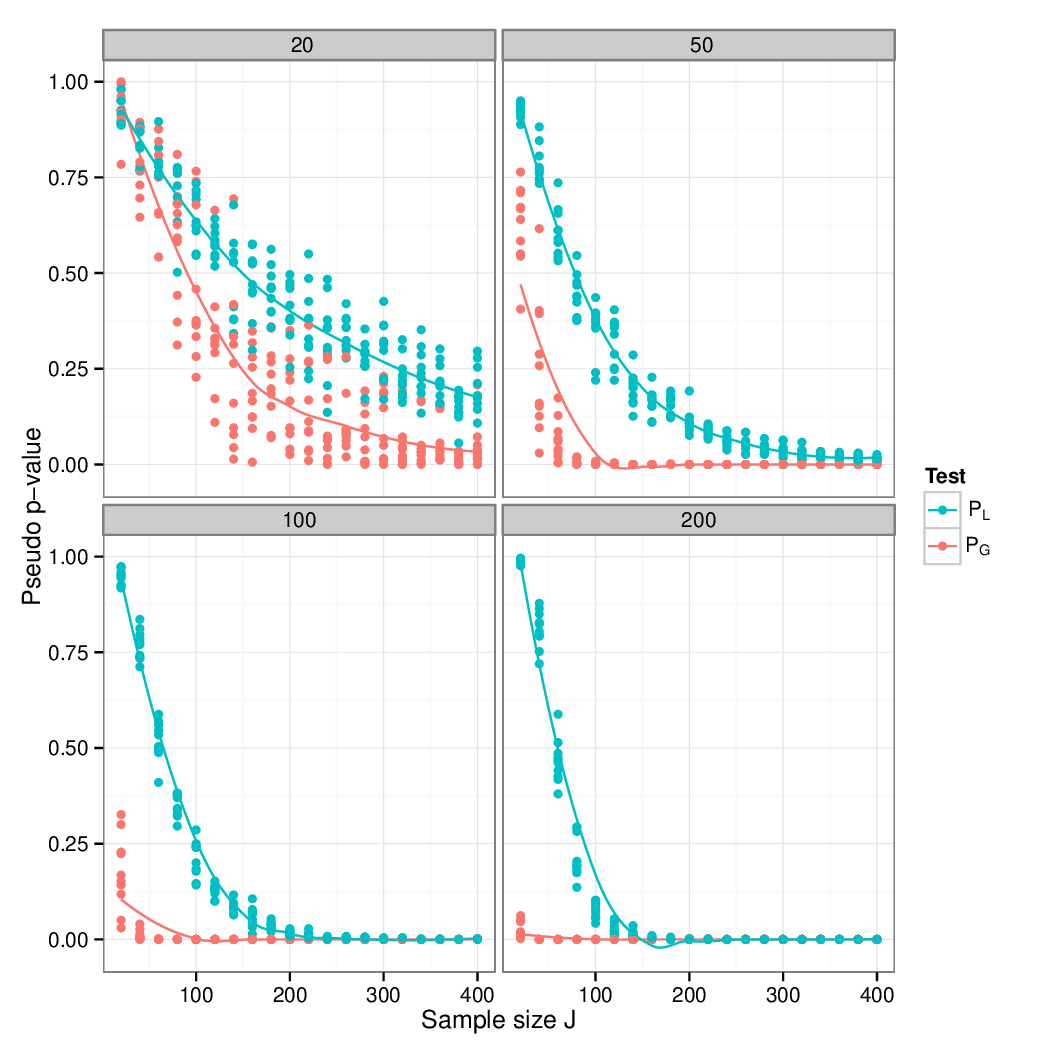}
\end{center}
\caption{{\bf Impact of sample number and size on detection of non-neutrality in the human gut data}. The figures show the pseudo p-values for neutrality for both the complete neutral model ($P_G$) and local community assembly ($P_L$). We generated ten replicate communities by sampling without replacement either 20, 50, 100 or 200 samples from those that had 1,000 reads or greater (247 in total) and from the selected samples we generated a fixed number of reads sampling with replacement. We increased read numbers from 20 individuals per sample to 400 inclusive in increments of 20. We then tested the subsampled communities for neutrality.}
\label{fig:Power}
\end{figure}

\begin{figure}[h]
\begin{center}
	\includegraphics[width=4in]{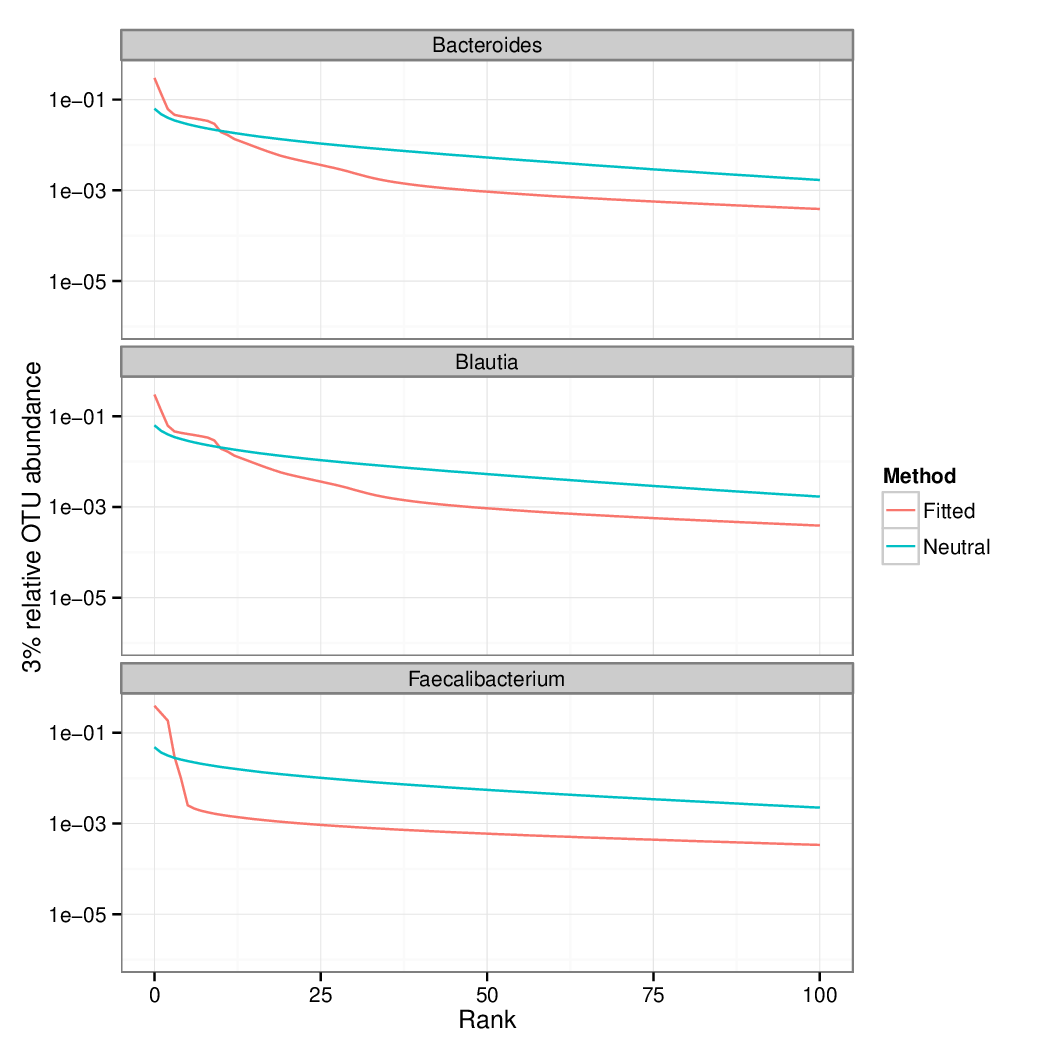}
\end{center}
\caption{{\bf Human gut metacommunity distributions}. The fitted metacommunity distributions (red line) and neutral metacommunity predictions (blue line)
as rank-abundance curves for three genera: Bacteroides, Blautia, and Faecalibacterium.}
\label{fig:Metapop}
\end{figure}

\begin{figure}[h]
\begin{center}
	\includegraphics[width=4in]{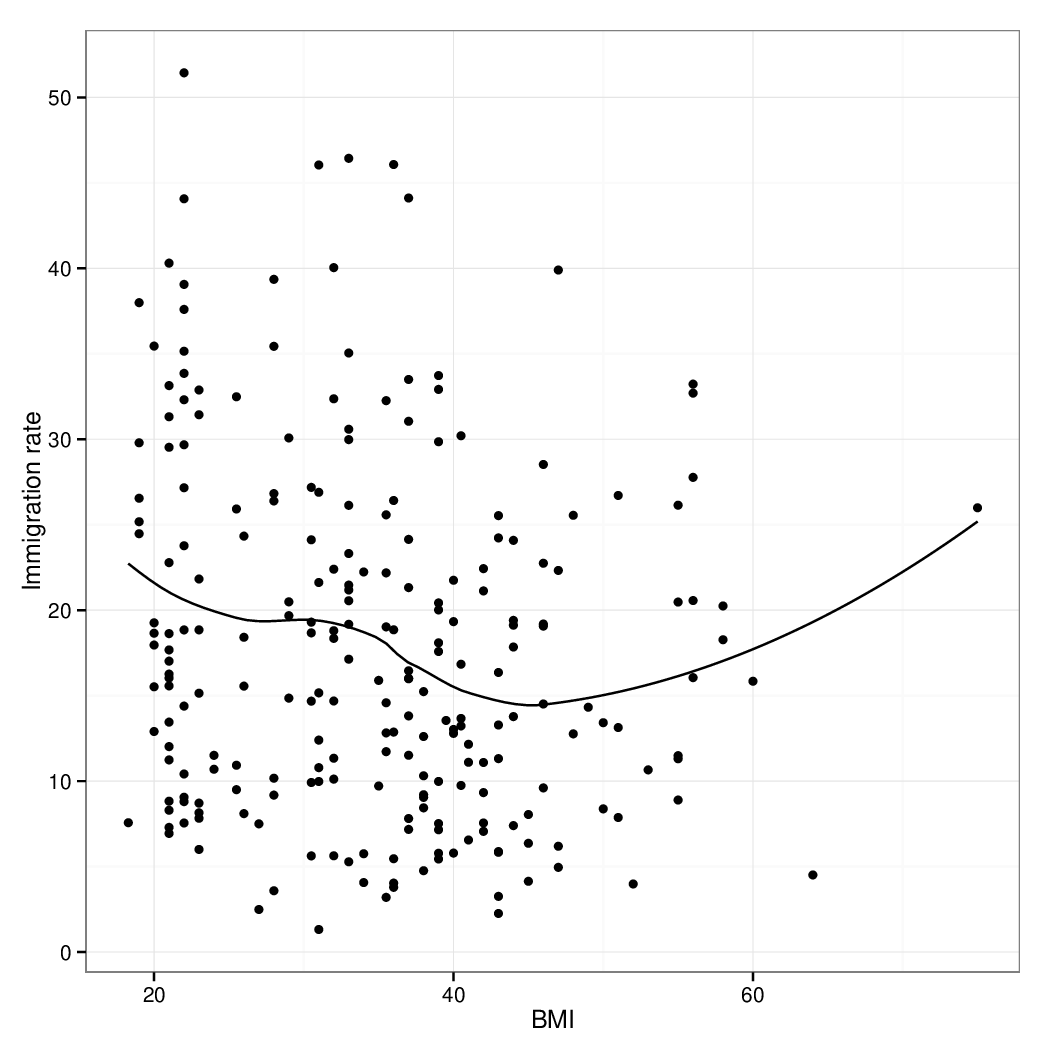}
\end{center}
\caption{{\bf Immigration rate vs. BMI}.  Median immigration rate for the family Ruminococcaceae determined by the UNTB-HDP model plotted against body mass index. A significant negative 
correlation is observed (p-value = 0.014 - Pearson's correlation).}
\label{fig:BMII}
\end{figure}


\clearpage
\newpage
\section*{Tables}

\begin{table}
\caption{{\bf Key ideas used in this paper.}}
\begin{tabular}{| p{.3\linewidth} | p{.7\linewidth} |}
\hline
Neutral model & A populaton model in which all types are \textit{functionally} equivalent\\
\hline
Unified Neutral Theory of Biodiversity (UNTB) & A discrete time stochastic model of an island-mainland metacommunity  proposed by Stephen Hubbell \cite{hubbell01}.   At each time step, one individual on the island dies, and is either replaced by the offspring of a randomly chosen individual on the island, or, with fixed probability, by the offspring of an individual chosen at random from the mainland.\\
\hline
Chinese Restaurant Process (CRP) & A discrete time stochastic model proposed by Davis Aldous \cite{aldous85} in which he imagines a Chinese restaurant with an unlimited number of tables.  At each time step, a new customer arrives, who will either choose a new table with a fixed probability $\theta$, or sit at an already occupied table with probability proportional to the number of individuals already seated at that table.  It is mathematically equivalent to Hoppe's urn \cite{Hoppe1984}, which generates samples from a Kingman coalescent with neutral mutations that occur at a fixed rate, and which always give rise to a new allelic type.\\
\hline
Dirichlet Process (DP) & A random variable taking value in the set of discrete probability distributions on a set $\cal{X}$, obtained by drawing random points in $\cal{X}$ according to a given probability measure $\mu$, and assigning these to the tables in a stationary Chinese Restaurant Process (thus, there are infinitely many customers seated at infinitely many tables), so that the probability of drawing a given point is equal to the proportion of customers seated at the corresponding table.\\
\hline
Hierarchical Dirichlet Process (HDP) & A Dirichlet Process for which the underlying measure $\mu$ is itself an instance of a Dirichlet Process.\\
\hline
\end{tabular}
\label{tab:Ideas} 
\end{table}

\begin{table}[h]
\caption{{\bf Fitting the UNTB-HDP model to human gut microbiota.}}
\begin{tabular}{|l|c|c|c|c|c|c|c|c|c|}
\hline
\multirow{2}{*}{Taxa} &  \multirow{2}{*}{$N$} & \multirow{2}{*}{$S$} & \multirow{2}{*}{$\tilde{J}$} & \multirow{2}{*}{$\theta$} & \multicolumn{3}{|c|}{$I_i$} & \multirow{2}{*}{$p_N$} & \multirow{2}{*}{$p_L$}\\
\cline{6-8}
& & & & & l & m & u & & \\
\hline
Bacteroidetes	& 231 & 569 & 596 & 148.6 & 1.5 & 5.5 & 13.7 & 0.0 (0.0) & 0.0 (0.0)\\	
\hline	
\hspace{1em}Bacteroidaceae & 208 & 224 & 506 & 51.4 & 0.7 & 3.3 & 7.6 & 0.0 (0.0) & 0.03 (0.0)\\			
\hspace{2em}Bacteroides & 208 & 224 & 506 & 51.4 &  0.7 & 3.3 & 7.6 & 0.0 (0.0) & 0.03 (0.0)\\
\hline
\hline
Firmicutes & 277 & 4770 & 1009 & 1382.3 & 21.4 & 44.8 & 81.0 & 0.0 (0.0) & 0.0 (0.0)\\
\hline			
\hspace{1em}Incertae Sedis XIV & 87 & 176 & 264 & 39.2 & 1.7 & 9.8 & 27.5 & 0.0 (0.0) & 0.05 (0.004)\\			
\hspace{2em}Blautia& 87 & 175 & 264 & 38.9 & 1.6 & 10.1 &  27.1 & 0.0 (0.0) & 0.06 (0.003)\\
\hspace{1em}Lachnospiraceae & 164 & 873 & 248 & 262.9 & 6.5 & 13.0 & 21.2 & 0.0 (0.0) & 0.0 (0.0)\\			
\hspace{1em}Ruminococcaceae & 239 & 1471 & 409 & 411.0 & 4.5 & 16.1 & 38.1 & 0.0 (0.0)& 0.0 (0.0)\\			
\hspace{2em}Faecalibacterium & 141 & 301 & 297 & 71.7 & 1.0 & 7.5 & 21.4 & 0.0 (0.0) & 0.004 (0.0)\\		
\hline\hline
\end{tabular}
\begin{flushleft}Results are given for 3\% OTUs at different levels, quantities given in the table are: $N$ -
the no. of samples with $> 150$ reads; $S$ - the number of 3\% OTUs; $\tilde{J}$ - the median sample size; $\theta$ -
the fitted biodiversity parameter; $I_i$ - the fitted immigration rates where l, m and u are the lower 2.5\%, median and upper 97.5\% quantiles respectively; $p_N$ - the proportion 
of simulated neutral samples exceeding the observed data likelihood; and $p_L$ - the proportion of 
simulated locally neutral samples exceeding the observed data likelihood. The figures in parantheses give pseudo p-values for the equivalent complete gut microbiome data set randomly sampled down to the same size as the individual taxa.
\end{flushleft}
\label{tab:Twins} 
\end{table}

\newpage
\clearpage

\begin{flushleft}
{\bf SI Appendix:  1) Large Population Limits for a Neutral Metacommunity and 2) Gibbs Sampling for the UNTB-HDP}
\end{flushleft}
\section{Large Population Limits for a Neutral Metacommunity}
\subsection{Summary and Outline}

Given the length and technical nature of this supplement, we will begin with a summary that outlines the results herein.   Our intent is to formulate a class of models that generalize Hubbell's formulation of the Unified Neutral Theory of Biodiversity and Biogeography (UNTB) and a number of variants that have appeared in the community ecology literature, whilst retaining the essential feature of neutrality.  Our inspiration in this are Cannings' models \cite{Cannings74}, which have become the standard in theoretical population genetics.  We discuss coalescent theory and these models in detail below, but in brief, a Cannings' model allows any reproduction law with discrete generations that keeps the total population size fixed, provided that relabeling the parents leaves the distribution of offspring unchanged.  More generally, we could consider models replacing fixed population sizes with density dependent population dynamics, as in \cite{Parsons+Quince07b}, \cite{Parsons2008,Parsons2010} and \cite{Parsons2012}, but this would have further lengthened and complicated this supplement. 

We formulate a mainland-island Cannings' model, in which the mainland has size $N_{0} = N$ and the islands have size $N_{i}$ that grow with $N$, but are approximately equal. 
We allow migration between any pair of island and mainland, and further allows mutations to give rise to new types on both island and mainland.  After collecting a few results regarding the reproduction law for a Cannings' model, we show in Section \ref{INTERMEDIATE}, provided that:
\begin{itemize}
\item the islands are asymptotically smaller than the mainland (in both census and effective population size; see the discussion below), 
\item migration between demes is rare (we assume that the probability that a migrant arrives in a local community is inversely proportional to the size of that community), and
\item the probability of multiple mergers is asymptotically smaller in $N$ than the rate of pairwise coalescence,
\end{itemize}
then Proposition \ref{INTERMEDIATEPROP} shows that if we rescale time proportionally to the effective population size of the islands (\ie we measure time so that one time step corresponds to $N_{e}$ generations) for large values of $N$, the population dynamics on the islands converge to the dynamics of Moran's infinitely many alleles model, with the migration rate from the mainland taking the place of the mutation rate in the population genetic model, and such that the type of all new mutants/migrants is drawn from the initial type distribution for the mainland (\ie the probability of migration between islands or novel mutations appearing on an island becomes vanishingly small as $N$ grows large, and can be completely ignored in the limit), and moreover, the composition of the mainland remains constant on this timescale - the dynamics are sufficiently slow that one cannot see changes when time is scaled according to the effective population size of the islands.  Moreover, this limit is independent of the specific reproduction law for the islands, provided it satisfies Cannings' conditions - indeed, we don't even need to assume the same law between islands. As a consequence of the identification of the islands' dynamics as a variation of the infinitely many alleles model, we can use previous results from theoretical population genetics to conclude that the stationary distribution for the islands is a Dirichlet Process, and that the composition of a sample is distributed according to Ewens' sampling formula.

In Section \ref{SLOW} we turn our attention to the mainland.  We first observe that for large values of time, the species distribution on the islands converge onto stationary processes governed by the Dirichlet process above.   We can then apply this with results from \cite{Ethier+Kurtz86} to obtain Proposition \ref{SLOWPROP}, which tells us that we need to rescale time according to the effective population size of the mainland (again, so that one time step corresponds to $N_{e}$ generations, but now $N_{e}$ for the mainland, which is substantially larger).  On this slow scale, the islands will essentially instantaneously arrive at their stationary state (an instant in this ``slow'' time scale is in fact an extremely long time in the natural ``intermediate'' time scale for the islands), whilst now the population on mainland follows the ``real'' infinitely many alleles model (with the actual mutation rate), and again, migrations from an island to the mainland become vanishingly rare as $N$ becomes large, and, as before, the stationary distribution is again a Dirichlet process, where each newly appearing genotype is assigned a label chosen uniformly at random from $[0,1]$ (thus the probability of two distinct mutations giving rise to the same type is 0).  In particular,  the islands have the Hierarchical Dirichlet Process for their stationary distribution: they are Dirichlet Processes in which the types are drawn from the underlying Dirichlet Process that describes the mainland.

\subsection{A Mainland-Island ``Cannings' Model''}

We begin by formulating a broad class of haploid models that includes Hubbell's Unified Neutral Theory of Biodiversity and Biogeography (UNTB) \cite{hubbell01}.  Our inspiration are Cannings' population genetic models \cite{Cannings74}, which use exchangeability as a general mathematical formulation of neutrality: random variables $\nu_{1},\ldots,\nu_{N}$ are \textit{exchangeable} if the random vectors $(\nu_{\pi(1)},\ldots,\nu_{\pi(N)})$ are equal in distribution for all permutations $\pi$ of $\{1,\ldots,N\}$.  Informally, the labels $1,\ldots,N$ are arbitrary, and can be changed without essentially changing the process.   In a Cannings' model, one assumes a fixed population of size $N$ and discrete generations; $\nu_{i}(n)$ is the number of offspring in the $n+1$\textsuperscript{st} generation of the $i$\textsuperscript{th} individual of the $n$\textsuperscript{th} generation.   $(\nu_{1},\ldots,\nu_{N})$  is assumed to be exchangeable and must satisfy
\[
	\sum_{i=1}^{N} \nu_{i} = N.
\]
Under suitable conditions on the higher moments (\nb as a consequence of exchangeability, we must have $\bbE\left[\nu_{i}\right] = 1$ for all $i$), one can show \cite{Mohle2001} that as $N \to \infty$ the frequency of types (here, the type of an individual is inherited from its ancestor in the initial population) and the genealogical process converge to the Wright-Fisher diffusion and Kingman's coalescent, respectively (relaxing the moment conditions leads to a $\Lambda$-coalescent limit for the genealogical process).  In particular, if $X^{(N)}_{i}(n)$ is the number of descendants alive in the \Th{n} generation of the \Th{i} ancestral individual in the \Th{0} generation, and $c_{N_{i}}$ is the coalescence probability, \ie the probability two individuals sampled without replacement from deme $i$ have the same parent,
\[
	c_{N} \defn \frac{\bbE\left[(\nu_{1})_{2}\right]}{N-1},
\]
where
\[ 
	(x)_{k} \defn x(x-1)\cdots(x-k+1)
\] 
is the \textit{falling factorial} or \textit{Pochhammer symbol}.  Then, \cite{Mohle2001} shows that 
 \[
 	\lim_{N \to \infty}  \frac{\bbE\left[(\nu_{1})_{3}\right]}{N\bbE\left[(\nu_{1})_{2}\right]} = 0
\]
is a necessary and sufficient condition for $X^{(N)}_{i}(\lfloor c_{N}^{-1} t \rfloor)$ to converge weakly\footnote{A family of random variables $\{X^{(N)}\}$ taking values in a space $S$ is said to \textit{converge weakly} to $X$ if 
\[
	\lim_{N \to \infty} \bbE[f(X^{(N)})] = \bbE[f(X)]
\]
for all $f \in C(S)$; the values $\bbE[f(X)]$ completely characterize the distribution of $X$.  Weak convergence is denoted by
\[
	X^{(N)} \Rightarrow X.
\]} as $N \to \infty$ to a Wright-Fisher diffusion, \ie to a diffusion process with probability density 
\[
	p(\by, t | \bx) \defn \bbP\left\{ \bX(t) \in \by+d\by \middle\vert \bX(0) = \bx\right\}
\]
satisfying the Kolmogorov backward equation
\[
	\frac{\partial p}{\partial t} = \frac{1}{2} \sum_{i=1}^{N} \sum_{j=1}^{N} x_{i}(\delta_{ij} - x_{j}) \frac{\partial p^{2}}{\partial x_{i} \partial x_{j}}.
 \]
 The quantity $c_{N}^{-1}$ has been referred to as the \text{coalescent effective population size}, and can be shown to generalize previously defined notions of an effective population size \cite{Sjodin2005}.
 
Here, we take our cues from the discussion of infinte-alleles models in \cite{Ethier+Kurtz86}, which we will closely follow, in formulating a ``Cannings' UNTB'' with migration and mutation.  As in previous models, we will assume a mainland, which supports a population of size $N_{0} = N$, together with a collection of islands labelled $i = 1,\dots,M$ which support populations of size $N_{i}$.  We will assume that the islands are all approximately the same size, and substantially smaller than the mainland; for Section \ref{INTERMEDIATE}, we will require $N_{i} \ll N_{0}$\footnote{We will write $a_{N} = o(b_{N})$, or $ a_{N} \ll b_{N}$, if
\[
	\lim_{N \to \infty} \frac{a_{N}}{b_{N}} = 0,
\]	
and use $a_{N} \asymp_{N} b_{N}$ to indicate that
\[
	\lim_{N \to \infty} \frac{a_{N}}{b_{N}} = 1.
\]
We will also write $a_{N} = \BigO{b_{N}}$ if there exists a constant $C$ such that
\[
	 a_{N} \leq C b_{N},
\]
for all $N$. 
}, whereas we will need to impose sharper estimates of the relative sizes in Section \ref{SLOW}.  In what follows, we will refer to the mainland and each of the islands as having $N_{0}$ or $N_{i}$ niches respectively, we will use the term deme when we are referring to a local community that can be either an island or the mainland, and will refer to \eg the individual in the \Th{j} niche in the \Th{i} deme. 

We will assume discrete generations, and that at each time step the current residents reproduce and are replaced by their offspring.  The \Th{j} individual has $\nu^{(N)}_{ij}$ offspring so that
\[
	\sum_{j=1}^{N_{i}} \nu_{ij} = N_{i},
\]
and model neutrality by assuming that each random vector $(\nu^{(N)}_{i1}(n),\ldots,\nu^{(N)}_{iN_{i}}(n))$ is exchangeable.  We further assume that $(\nu^{(N)}_{i1}(n),\ldots,\nu^{(N)}_{iN_i}(n))$ is independent of $(\nu^{(N)}_{j1}(m),\ldots,\nu^{(N)}_{jN_{j}}(m))$ unless $i = j$ and all $m = n$.  Following \cite{Mohle2001}, we define
\[
	c_{N_{i}} \defn \frac{\bbE\left[(\nu_{i1})_{2}\right]}{N_{i}-1},
\]	
for $i = 0,\ldots,M$, and assume the analogue of M\"ohle's condition:
\begin{equation}\label{MOHLECOND}
	\lim_{N_{i} \to \infty}  \frac{\bbE\left[(\nu_{i1})_{3}\right]}{N_{i}\bbE\left[(\nu_{i1})_{2}\right]} = 0,
\end{equation}
which has the following consequence \cite{Mohle+Sagitov2003}:

\begin{lem}
Assume \eqref{MOHLECOND}.   Then,
\[
	\lim_{N_{i} \to \infty} c_{N_{i}} = 0,
\]
and
\[
 	\lim_{N_{i} \to \infty}  \frac{\bbE\left[(\nu_{1})_{2}(\nu_{2})_{2}\right]}{c_{N_{i}}} = 0.
\]
\end{lem}

We will further assume that there exists $a_{N}$ such that
\begin{equation}\eqnlabel{GAMMADEF}
	\lim_{N \to \infty} \frac{c_{N_{i}}}{a_{N}} = \begin{cases}	
		\gamma_{i} & \text{if $i > 0$, and}\\
		0 & \text{otherwise.}
	\end{cases} 
\end{equation}
which formalises the notion that the populations on the islands are all of the same order of magnitude (their effective population sizes are asymptotically proportional $c_{N_{i}} \sim \gamma_{i} a_{N}$) and asymptotically smaller than the mainland ($a_{N} \ll c_{N}$).

We will further assume that each individual has a type, which is a label in $[0,1]$, which we think of as a probability space with the uniform (Lebesgue) measure $\lambda$.  The labels are more of a mathematical convenience for tracking ancestries, and have no effect on fitness, so we could equally well take labels in any compact Polish space $\mathfrak{X}$ that is equipped with a probability measure $\gamma(dx)$.   We write $X_{ij}(n) \in [0,1]$ for the type of the individual in the \Th{j} niche of the \Th{i} deme in generation $n$ -- the labels are inherited from the parent, except when an individual is subject to mutation at birth. We discuss the processes of reproduction and mutation below.  The state of the \Th{i} deme in the \Th{n} generation is conveniently represented by an atomic probability measure on $[0,1]$, 
\[
	G^{(N)}_{i}(n) = \frac{1}{N_{i}}\sum_{j=1}^{N_{i}} \delta_{X_{ij}(n)},
\]
where $\delta_{X_{ij}(t)}$ is the Dirac point mass at $X_{ij}(t)$, and the superscript $(N)$ emphasizes the dependence on the ``system size'' $N$, \ie for any subset $A \subseteq [0,1]$, $G^{(N)}_{i}(n)(A)$ is the number of individuals in the \Th{i} deme with a type in the set $A$.  We write $\bG^{(N)}(n) = G^{(N)}_{0}(n) \otimes \cdots \otimes G^{(N)}_{M}(n)$\
for the product measure,
\[
	\bG^{(N)}(n)(A) = G^{(N)}_{0}(n)(A) \cdots G^{(N)}_{M}(n)(A).
\]

Given a measure $\mu$ and a continuous function $f$ on $[0,1]$,   we will use the shorthand
\[
	\langle f, \mu \rangle \defn \int f(x)\, \mu(dx)
\]
for the integral.   More generally, if $f \in C([0,1]^{M+1})$, then
\[
	\langle f, \mu_{0} \otimes \cdots \otimes \mu_{M} \rangle \defn \int f(x_{0},\ldots,x_{M})\, \mu_{0}(dx_{0}) \cdots \mu_{M}(dx_{M}).
\]	
By definition, we have
\[
	\langle f, G^{(N)}_{i}(n) \rangle = \frac{1}{N_{i}}\sum_{j=1}^{N_{i}} f(X_{ij}(n)).
\]

We model migration by assuming that with probability $c_{N_{i}} \frac{\varpi_{i}}{2}$ (the factor of $\frac{1}{2}$ is to maintain consistency of notation with the cited population genetics literature), a given individual in the $n+1$\textsuperscript{st} generation is replaced by the migrant offspring of a parent chosen uniformly at random from the entire metapopulation, \ie we assume a parent of type $X_{pq}(n)$, where the $p$ and $q$ are drawn uniformly from $\{0,\ldots,M\}$ and $\{1,\ldots,N_{p}\}$, respectively.  Thus, the average number of migrants to a given island is asymptotically independent of $N$; this is a weak migration limit.   Equivalently, the parent is drawn from the \textit{metapopulation} measure,
\begin{equation}\label{METAPOP}
	G^{(N)}(n) \defn \frac{1}{\sum_{k=0}^{M} N_{k}} \sum_{i=0}^{M} N_{i} G^{(N)}_{i}(n).
\end{equation}
 
Finally, we allow for the possibility that individuals mutate after birth; we assume that there is a probability measure $P^{(N)}$ such that the offspring of a parent with type $x \in [0,1]$ mutates to a type in $A \subseteq [0,1]$ with probability $P^{(N)}(x,A)$.   Define an operator $Q^{(N)}$ on  $C([0,1])$ by 
\[
	(Q^{(N)} f)(x) = \int_{0}^{1} f(y) P^{(N)}(x,dy).
\]
Then, for all $f \in C([0,1])$, we define
\begin{multline}\label{QI}
	(Q^{(N)}_{i} f)(x) \defn \bbE\left[f(X_{ij}) \middle\vert \bG^{(N)}(n), \text{parent of type $x$}\right] \\ 
		= (1-c_{N_{i}} \frac{\varpi_{i}}{2}) (Q^{(N)}f)(x) + c_{N_{i}} \frac{\varpi_{i}}{2} \int (Q^{(N)}f)(y) G^{(N)}(n)(dy)
\end{multline}
and 
\begin{equation}\label{BI}
	(B^{(N)}_{i}f)(x) \defn \frac{\varpi_{i}}{2} \left(\int f(y) G^{(N)}(n)(dy) - f(x)\right).
\end{equation}
While it may at first appear unusual, this notation will greatly simplify subsequent calculations.

We will assume mutation is weak:  
\[
	B \defn \lim_{N \to \infty} c_{N}^{-1} (I - Q^{(N)})
\]
exists and $B$ is a bounded operator.  Thus, for any set $A \subseteq [0,1]$, the probability that the offspring has a type in $A$ approaches 1 as $N \to \infty$, if the parent has a type in $A$, and approaches 0 otherwise.   Here, $c_{N}$ is the coalescent effective population size for the mainland, and we are making the standard assumption that mutation rates scale like the reciprocal of the effective population size.  For the sake of clarity in the arguments that follow, we emphasize that our assumptions entail that
\[
	Q^{(N)}_{i} =  I + c_{N_{i}} B^{(N)}_{i} + c_{N} B + o(c_{N}).
\]

One can consider many forms for the operator $B$; the operator
\[
	(B^{(L)} f)\left(\frac{i}{L}\right)= \frac{\theta}{L-1} \sum_{j = 1}^{L} \left(f\left(\frac{j}{L}\right) - f\left(\frac{i}{L}\right)\right)
\]
corresponds to the classical population genetic models, in which the number of possible types is discrete and finite (here, there are $L$) and mutation is symmetric (\ie the offspring of an individual have the same type as their parent with probability $1-\frac{\theta}{N}$, and mutate to any other type with probability $\frac{\theta}{N(L-1)}$).  Since the labels are arbitrary, they can be assumed to be chosen from the set $\left\{\frac{1}{L},\frac{2}{L},\ldots,1\right\}$.  Now, as $L \to \infty$, $B^{(L)}$ converges to the operator 
\[
	(B f)(x) = \frac{\theta}{2} \int_{0}^{1} f(y)\, dy - f(x) = \theta (\langle f, \lambda\rangle - f(x)),
\]
which corresponds to the infinitely many alleles model; the probability that two mutations give rise to the same type is 0.  We will henceforth assume $B$ is of this form.

\begin{rem}
Although we have formulated the community dynamics in discrete time, we could equally well consider a continuous time Markov process $\tilde{G}^{(N)}_{i}(t)$ in which disturbances happen at some rate $D$; in the latter case, we consider the embedded Markov chain: if disturbances happen at random times $\tau_{1},\tau_{2},\ldots$, then the embedded chain is the process $G^{(N)}_{i}(n) \defn \tilde{G}^{(N)}_{i}(\tau_{n})$.  The limiting (continuous time) process as $N \to \infty$ is the same for both $G^{(N)}$ and $\tilde{G}^{(N)}_{i}$.
\end{rem}

In the next  section, we will consider the limiting behaviour as first $N$ and then $L$ are taken to infinity.  We will see that under moment assumptions corresponding to those in \cite{Mohle2001}, all of these models converge to the same limiting process.  First, however, we illustrate how Hubbell's original UNTB is an example of our class of models.

\begin{exmp}[Hubbell's UNTB]
In Hubbell's original model \cite{hubbell01}, only a single individual is replaced in each deme at at each time step.  We then have $\nu_{0i}$ takes values in $\{0,1\}$, with 
\[
	\bbP\{\nu_{0i} = 1\} = m.
\]
We then have that the remaining offspring numbers are either 
\[
	(\nu_{i1},\ldots,\nu_{iN_{i}})  = (1,\ldots,1,0,1,\ldots,1)
\]
(the vector with \Th{i} entry 0, and all others 1), if $\nu_{0i} = 1$, and is 
\[
	(\nu_{i1},\ldots,\nu_{iN_{i}})  = (1,\ldots,1,0,1,\ldots,1,2,1,\ldots,1)
\]
(the vector with \Th{i} entry 0 and \Th{j} entry 2 for some $i \neq j$), if $\nu_{0i} = 0$, with conditional probabilities equal to $\frac{1}{ N_{i}}$ and $\frac{1}{ N_{i}(N_{i}-1)}$, respectively (and thus the $\nu_{ij}$ are exchangeable, given $\nu_{i0}$).  

For this model, we have 
\[
	c_{N_{i}} = \frac{2}{N_{i}(N_{i}-1)},
\]
whereas by definition, $(\nu_{i1})_{3} = 0$, so \eqref{MOHLECOND} holds.

In Hubbell's model, immigrants are always from the mainland, which is assumed to have a fixed, stationary distribution (usually taken so that samples from the mainland are distributed according to Ewens' sampling formula \cite{ewens72}) , and no mutations are assumed to occur on the islands.  We will not need to make these assumptions, but will instead derive them (in the limit as $N \to \infty$) as a consequence of the relative size of the mainland and the islands.  
\end{exmp}

\begin{exmp}[``Wright-Fisher'' UNTB]
We can regard Hubbell's UNTB as a community analogue of the discrete Moran model.  We could similarly define a community analogue to the Wright-Fisher model by assuming that the vector $(\nu_{i1},\ldots,\nu_{iN_{i}})$ follows a multinomial distribution with parameters $N_{i}$ and $\left(\frac{1}{N_{i}},\ldots,\frac{1}{N_{i}}\right)$, \ie for each $i$:
\[
	\bbP\left\{(\nu_{i1},\ldots,\nu_{iN_{i}}) = (k_{1},\ldots,k_{N_{i}}) \right\} 
		= \frac{N_{i}!}{k_{1}! \cdots k_{N_{i}}!} \left(\frac{1}{N_{i}}\right)^{k_{1}} \cdots \left(\frac{1}{N_{i}}\right)^{k_{N_{i}}}.
\] 
Here, $c_{N_{i}} = \frac{1}{N_{i}}$, whereas $\bbE\left[(\nu_{i1})_{3}\right] = \frac{1}{N_{i}}^{2}$.
\end{exmp}

\begin{exmp}
We briefly note that it is possible to have $c_{N_{i}} \equiv 0$, by assuming that $(\nu_{i1},\ldots,\nu_{iN_{i}}) = (1,\ldots,1)$ with probability 1 (a trivial case that we will ignore), whereas it need not be the case that
\[
	\lim_{N_{i} \to \infty} c_{N_{i}} = 0
\]
if \eqref{MOHLECOND} is violated:  if we assume that with probability $\frac{1}{N_{i}}$, $\nu_{ij}  = N_{i}$ and $\nu_{ik}  = 0$ for all $k \neq j$, then $c_{N_{i}} \equiv 1$.
\end{exmp}


\subsection{Preliminaries Considering Exchangeable Variables}

It is well known \cite{Kingman82b} that 
\[
	\frac{(N_{i})_{j}}{(N_{i})_{k}}\bbE\left[(\nu_{i1})_{k_{1}}\cdots(\nu_{ij})_{k_{j}}\right],
\]
where $j,k_{1},\ldots,k_{j} \in \bbN$ and $k \defn k_{1} + \cdots + k_{j}$, is the probability that $k$ individuals, sampled uniformly at random without replacement from the \Th{i} deme have exactly $j$ parents in the previous generation, \nb exchangeability implies that 
\[
	\frac{(N_{i})_{j}}{(N_{i})_{k}}\bbE\left[(\nu_{i1})_{k_{1}}\cdots(\nu_{ij})_{k_{j}}\right] = \frac{(N_{i})_{j}}{(N_{i})_{k}}\bbE\left[(\nu_{i\pi(1)})_{k_{1}}\cdots(\nu_{i\pi(j)})_{k_{j}}\right]
\]
for any permutation $\pi$ of $\{1,\ldots,N_{i}\}$, so that these probabilities only depend on $j$, $k$, and the unordered list of values $k_{1},\ldots,k_{j}$.  In \cite{Mohle2001}, we find the following monotonicity result for these probabilities:

\begin{lem}\label{MOHLELEM}
Let $j \geq l$, $k_{1} \geq m_{1},\dots,k_{l} \geq m_{l}$, and $m \defn m_{1} + \cdots + m_{l}$.  Then,
\[
	\frac{(N_{i})_{j}}{(N_{i})_{k}}\bbE\left[(\nu_{i1})_{k_{1}}\cdots(\nu_{ij})_{k_{j}}\right] 
	\leq \frac{(N_{i})_{l}}{(N_{i})_{m}}\bbE\left[(\nu_{i1})_{m_{1}}\cdots(\nu_{ij})_{m_{l}}\right].
\]
\end{lem}

\begin{rem}
In particular, in conjunction with Lemma \ref{MOHLELEM}, we have
 \begin{equation}\label{MOHLECONS1}
	\frac{(N_{i})_{j-1}}{(N_{i})_{j}}\bbE\left[(\nu_{i1})_{2}\nu_{i2}\cdots\nu_{ij-1}\right] \leq c_{N_{i}},
\end{equation}
(and, by exchangeability, whenever at least one $k_{i} \geq 2$) and 
\begin{equation}\label{MOHLECONS2}
	\frac{(N_{i})_{j}}{(N_{i})_{k}}\bbE\left[(\nu_{i1})_{k_{1}}\cdots(\nu_{ij})_{k_{j}}\right]  = o(c_{N_{i}})
\end{equation}
whenever $k_{q},k_{r} \geq 2$ for at least two distinct indices $q$, $r$ or $k_{q} \geq 3$ for some index $q$.
\end{rem}

\begin{rem}
In particular, in \eqref{MOHLECONS2}, $\frac{(N_{i})_{j}}{(N_{i})_{k}}\bbE\left[(\nu_{i1})_{k_{1}}\cdots(\nu_{ij})_{k_{j}}\right]$ is always smaller than one of
\[
	\frac{(N_{i})_{1}}{(N_{i})_{3}}\bbE\left[(\nu_{i1})_{3}\right] 
\]
or 
\[
	\frac{(N_{i})_{2}}{(N_{i})_{4}}\bbE\left[(\nu_{i1})_{2}(\nu_{i2})_{2}\right]. 
\]
In what follows, all terms $o(c_{N_{i}})$ will be of order at most equal to one of these two quantities (which are the probability of three individuals sampled at random having the same parent in the previous generation, or a sample of four individuals consisting of two pairs of descendants of two distinct parents, respectively) or will be of order less than or equal to $\frac{c_{N_{i}}}{N_{i}}$.  This will be very important when we consider the long timescale.
\end{rem}

We will have use of some general relations between exchangeable random variables in the sequel:

\begin{lem}
For all $j > 1$
\[
	\bbE\left[\nu_{i1}\cdots\nu_{ij-1}\right] - \bbE\left[\nu_{i1}\cdots\nu_{ij}\right] = (j-1)\frac{(N_{i})_{j-1}}{(N_{i})_{j}} \bbE\left[(\nu_{i1})_{2}\cdots\nu_{ij-1}\right].
\]
\end{lem}

\begin{proof}
We begin by observing that 
\begin{align*}
	N_{i} \bbE\left[\nu_{i1}\cdots\nu_{ij-1}\right] &= \bbE\left[N_{i}\nu_{i1}\cdots\nu_{ij-1}\right]\\
	&= \bbE\left[(\nu_{i1} + \cdots + \nu_{iN_{i}})\nu_{i1}\cdots\nu_{ij-1}\right]\\
	&= \bbE\left[\sum_{k=1}^{N_{i}} \nu_{i1}\cdots\nu_{ij-1}\nu_{ik}\right]\\
	&= \bbE\left[\sum_{k=1}^{j-1} \nu_{i1}\cdots\nu_{ij-1}\nu_{ik}+ \sum_{k=j}^{N_{i}} \nu_{i1}\cdots\nu_{ij-1}\nu_{ik}\right]\\
	&= \sum_{k=1}^{j-1} \bbE\left[\nu_{i1}\cdots\nu_{ij-1}\nu_{ik}\right]+ \sum_{k=j}^{N_{i}} \bbE\left[\nu_{i1}\cdots\nu_{ij-1}\nu_{ik}\right]\\
	&= \sum_{k=1}^{j-1} \bbE\left[\nu_{ik}^{2} \prod_{{l=1} \atop {l \neq k}}^{j-1} \nu_{il}\right]+ \sum_{k=j}^{N_{i}} \bbE\left[\nu_{i1}\cdots\nu_{ij-1}\nu_{ik}\right]\\
\intertext{and thus, exploiting the exchangeability of the $\nu_{ij}$,}
	&= (j-1)\bbE\left[\nu_{i1}^{2}\cdots\nu_{ij-1}\right] + (N_{i}-j+1)\bbE\left[\nu_{i1}\cdots\nu_{ij}\right]. 
\end{align*}
On the other hand, 
\[
	N_{i} \bbE\left[\nu_{i1}\cdots\nu_{ij-1}\right] 
		= (j-1)\bbE\left[\nu_{i1}\cdots\nu_{ij-1}\right] + (N_{i}-j+1)\bbE\left[\nu_{i1}\cdots\nu_{ij-1}\right].
\]
Equating the two sides and subtracting, we get 
\[
	(N_{i}-j+1)\left(\bbE\left[\nu_{i1}\cdots\nu_{ij-1}\right] - \bbE\left[\nu_{i1}\cdots\nu_{ij}\right]\right) = 
		(j-1)\left(\bbE\left[\nu_{i1}^{2}\cdots\nu_{ij-1}\right] - \bbE\left[\nu_{i1}\cdots\nu_{ij-1}\right]\right).
\]	
The result follows.
\end{proof}

\begin{rem}\label{DIFFERENCE}
In conjunction with \eqref{MOHLECONS1}, the lemma tells us that for all $j > 1$, 
\[
	\bbE\left[\nu_{i1}\cdots\nu_{ij-1}\right] - \bbE\left[\nu_{i1}\cdots\nu_{ij}\right] = \BigO{c_{N_{i}}},
\]
and thus, 
\[
	\bbE\left[\nu_{i1}\cdots\nu_{iq}\right] - \bbE\left[\nu_{i1}\cdots\nu_{ir}\right] = \BigO{c_{N_{i}}},
\]
for any $q < r$.
\end{rem}

Next, we observe that

\begin{lem}\label{ONE}
For all $j$,
\[
	\bbE\left[\nu_{i1}\cdots\nu_{ij}\right] = 1- {j \choose 2} \frac{(N_{i})_{j-1}}{(N_{i})_{j}} \bbE\left[(\nu_{i1})_{2}\cdots\nu_{ij-1}\right] - o(c_{N_{i}}).
\]
\end{lem}

\begin{proof}
This is a consequence of the identity
\[
	(N_{i})_{j} = (\nu_{i1} + \cdots + \nu_{iN_{i}})_{j} = \sum_{j_{1}+\cdots+j_{N_{i}}=j} \frac{j!}{j_{1}! \cdots j_{N_{i}}!} (\nu_{i1})_{j_{1}}\cdots(\nu_{iN_{i}})_{j_{N_{i}}},
\]
where we assume $0! = 1$ for ease of notation, and we assume that most of the $j_{i}$ are equal to zero.  Equivalently, if we only consider non-zero values,
\[
	(N_{i})_{j} = \sum_{m=1}^{j} \sum_{{n_{1},\ldots,n_{m}} \atop{\text{distinct}}} \sum_{k_{1} + \cdots + k_{m} = k}  \frac{j!}{k_{1}! \cdots k_{m}!} 
		(\nu_{in_{1}})_{k_{1}}\cdots(\nu_{in_{m}})_{k_{m}}.
\]
Taking expectations on both sides, and using the exchangeability of $(\nu_{i1},\ldots,\nu_{iN_{i}})$, we have
\[
	(N_{i})_{j} = \sum_{m=1}^{j} \sum_{{n_{1},\ldots,n_{m}} \atop{\text{distinct}}} \sum_{k_{1} + \cdots + k_{m} = k}  \frac{j!}{k_{1}! \cdots k_{m}!} 
		\bbE\left[(\nu_{i1})_{k_{1}}\cdots(\nu_{im})_{k_{m}}\right].
\]
Now, observe that the expected value of the summand is independent of the choice of the values $n_{1},\ldots,n_{m}$, that can be chosen in $N \choose m$ ways.  Moreover, the expectation $\bbE\left[(\nu_{i1})_{k_{1}}\cdots(\nu_{im})_{k_{m}}\right]$ remains unchanged under permutations, and thus are all equal to 
\[
	\bbE\left[(\nu_{i1})_{\tilde{k}_{1}}\cdots(\nu_{im})_{\tilde{k}_{m}}\right],
\]
where $\tilde{k}_{1} \geq \tilde{k}_{2} \geq \cdots \geq \tilde{k}_{m}$ are the values $k_{1},\ldots,k_{m}$ listed in decreasing order.  If we let $a_{p}$ be the number of indices $q$ such that $k_{q} = p$, 
\[
	a_{p} = \#\{ q : k_{q} = p \},
\]
then
\begin{multline*}
	\sum_{m=1}^{j} \sum_{{n_{1},\ldots,n_{m}} \atop{\text{distinct}}} \sum_{k_{1} + \cdots + k_{m} = k}  \frac{j!}{k_{1}! \cdots k_{m}!} 
		\bbE\left[(\nu_{i1})_{k_{1}}\cdots(\nu_{im})_{k_{m}}\right]\\
	= \sum_{m=1}^{j} \sum_{{\tilde{k}_{1} + \cdots + \tilde{k}_{m} = k} \atop {\tilde{k}_{1} \geq \tilde{k}_{2} \geq \cdots \geq \tilde{k}_{m}}}  \frac{j!}{\tilde{k}_{1}! \cdots \tilde{k}_{m}!} 
		\frac{m!}{a_{1}! \cdots a_{j}!}{N \choose m} \bbE\left[(\nu_{i1})_{\tilde{k}_{1}}\cdots(\nu_{im})_{\tilde{k}_{m}}\right],
\end{multline*}
so that, simplifying and dividing through by $(N)_{j}$, we have
\begin{align*}
	1 &= \sum_{m=1}^{j} \sum_{{\tilde{k}_{1} + \cdots + \tilde{k}_{m} = k} \atop {\tilde{k}_{1} \geq \tilde{k}_{2} \geq \cdots \geq \tilde{k}_{m}}}  \frac{j!}{\tilde{k}_{1}! \cdots \tilde{k}_{m}!} 
		\frac{1}{a_{1}! \cdots a_{j}!} \frac{(N)_{m}}{(N)_{j}} \bbE\left[(\nu_{i1})_{\tilde{k}_{1}}\cdots(\nu_{im})_{\tilde{k}_{m}}\right]\\
		&= \bbE\left[\nu_{i1}\cdots\nu_{ij}\right] + {j \choose 2} \frac{(N_{i})_{j-1}}{(N_{i})_{j}} \bbE\left[(\nu_{i1})_{2}\cdots\nu_{ij-1}\right] + o(c_{N_{i}}),
\end{align*}
where, using \eqref{MOHLECONS2}, we have truncated after the two highest order terms in the sum.
\end{proof}

We conclude this section with a final observation, 

\begin{lem}\label{TWO}
Let $j > 1$.  Then, 
\[
	\frac{(N_{i})_{j}}{(N_{i})_{j+1}}\bbE\left[(\nu_{i1})_{2}\nu_{i2}\cdots\nu_{ij}\right]  
		= \frac{(N_{i})_{j-1}}{(N_{i})_{j}}\bbE\left[(\nu_{i1})_{2}\nu_{i2}\cdots\nu_{ij-1}\right] + o(c_{N_{i}}). 
\]
\end{lem}

\begin{proof}
Again exploiting exchangeability, we see that
\begin{align*}
	(N_{i}-j+&1)\bbE\left[(\nu_{i1})_{2}\nu_{i2}\cdots\nu_{ij}\right] \\
		&= \bbE\left[(\nu_{i1})_{2}\nu_{i2}\cdots\nu_{ij-1}\nu_{ij}\right] + \bbE\left[(\nu_{i1})_{2}\nu_{i2}\cdots\nu_{ij-1}\nu_{ij+1}\right] + \cdots 
			+ \bbE\left[(\nu_{i1})_{2}\nu_{i2}\cdots\nu_{ij-1}\nu_{iN_{i}}\right]\\
		&= \bbE\left[(\nu_{i1})_{2}\nu_{i2}\cdots\nu_{ij-1}(\nu_{ij}+\cdots+\nu_{iN_{i}})\right]\\
		&= \bbE\left[(\nu_{i1})_{2}\nu_{i2}\cdots\nu_{ij-1}(N_{i}-\nu_{i1}-\cdots-\nu_{ij-1})\right]\\
		&= \bbE\left[(\nu_{i1})_{2}\nu_{i2}\cdots\nu_{ij-1}(N_{i}-j+2-(\nu_{i1}-2) - (\nu_{i2}-1)-\cdots-(\nu_{ij-1}-1))\right]\\
		&= (N_{i}-j+2)\bbE\left[(\nu_{i1})_{2}\nu_{i2}\cdots\nu_{ij-1}\right] - \bbE\left[(\nu_{i1})_{3}\nu_{i2}\cdots\nu_{ij-1}\right] \\
		&\qquad - \bbE\left[(\nu_{i1})_{2}(\nu_{i2})_{2}\nu_{i3}\cdots\nu_{ij-1}\right] - \cdots - \bbE\left[(\nu_{i1})_{2}\nu_{i2}\cdots(\nu_{ij-1})_{2}\right].
\end{align*}
In particular, dividing both sides by $(N_{i}-j+1)(N_{i}-j+2)$, we have 
\begin{multline*}
	\frac{(N_{i})_{j}}{(N_{i})_{j+1}}\bbE\left[(\nu_{i1})_{2}\nu_{i2}\cdots\nu_{ij}\right] = \frac{(N_{i})_{j-1}}{(N_{i})_{j}}\bbE\left[(\nu_{i1})_{2}\nu_{i2}\cdots\nu_{ij-1}\right] 
	- \frac{(N_{i})_{j-1}}{(N_{i})_{j+1}}\bbE\left[(\nu_{i1})_{3}\nu_{i2}\cdots\nu_{ij-1}\right] \\
	- \frac{(N_{i})_{j-1}}{(N_{i})_{j+1}}\bbE\left[(\nu_{i1})_{2}(\nu_{i2})_{2}\nu_{i3}\cdots\nu_{ij-1}\right] - \cdots 
	- \frac{(N_{i})_{j-1}}{(N_{i})_{j+1}}\bbE\left[(\nu_{i1})_{2}\nu_{i2}\cdots(\nu_{ij-1})_{2}\right]
\end{multline*}
and the result again follows by  \eqref{MOHLECONS2}.
\end{proof}

\begin{rem}
Iterating the previous lemma, we see that 
\[
	\frac{(N_{i})_{j}}{(N_{i})_{j+1}}\bbE\left[(\nu_{i1})_{2}\cdots\nu_{ij}\right]  = \cdots = \frac{\bbE\left[(\nu_{i1})_{2}\right]}{N_i-1} + o(c_{N_{i}}) = c_{N_{i}} + o(c_{N_{i}}). 
\]
\end{rem}

\subsection{Convergence to a Limit}\label{INTERMEDIATE}

We will be interested in weak limits of the random measures $\bG^{(N)}(n)$ in two time-scales determined by $N$, a ``slow-time'' process, $\bG^{(N)}(\lfloor c_{N}^{-1} t \rfloor)$, and an ``intermediate-time'' process $\bG^{(N)}(\lfloor a_{N}^{-1} t \rfloor)$, where $t > 0$ is a continuous time variable, and we will consider the limits as $N \to \infty$.

Our principal tool in doing this is the generator of $\bG^{(N)}(n)$, an operator on $C(\mathscr{P}([0,1])^{M+1})$ defined by
\[
	(\mathcal{G}_{N} F)(\bmu) = \mathbb{E}\left[ F(\bG^{(N)}(n+1)) \middle\vert \bG^{(N)}(n) = \bmu \right] - F(\bmu).
\]
Knowing $(\mathcal{G}_{N} F)(\bmu)$ for all $F \in C(\mathscr{P}([0,1])^{M+1})$ and all $\bmu \in \mathscr{P}([0,1])^{M+1}$ completely characterizes the transition probabilities of $\bG^{(N)}$, and thus, together with the initial value $\bG^{(N)}(0)$, allow us to characterize the process (although not necessarily the limit, see \eg \cite{Ethier+Kurtz86}).

Our limiting processes are continuous, rather than discrete time random variables, but also have associated generators; in general, if $\bH(t)$ is a continuous time process taking values in $\mathscr{P}([0,1])^{M+1}$ and $F \in C(\mathscr{P}([0,1])^{M+1})$, then $\bH(t)$ has generator $\mathcal{H}$:
\[
	(\mathcal{H}F)(\bmu) = \lim_{h \to 0^{+}} \frac{ \mathbb{E}\left[ F(\bH(t+h)) \middle\vert \bH(t) = \bmu \right] - F(\bmu)}{h},
\]
with domain $\mathscr{D}(\mathcal{H})$, consisting of all functions $F$ for which the limit exists.  

The notion of a generator simultaneously generalizes the transition matrix, master equation, and diffusion equations of classical probability.   The typical proof of convergence  proceeds by first showing that a limit exists, then characterizing the limit by first determining the limit of the generators, and finally showing that given the initial conditions (via a distribution from which they are drawn), there is a unique process with that generator (\eg \cite{Ethier+Kurtz86} is a standard reference).

\begin{rem}
Note that $(\mathcal{H}F)(\bmu)$ is the right-hand derivative of $\mathbb{E}\left[ F(\bH(t+h)) \middle\vert \bH(t) = \bmu \right]$ at $t=0$.  In particular, if the generator vanishes,
then $\mathbb{E}\left[ F(\bH(t)) \middle\vert \bH(0) = \bmu \right] = F(\bmu)$ for all $t > 0$, and all $F$, and the process $\bH(t) \equiv \bmu$ is constant.   This will be important when we come to consider the limit on the intermediate time scale.
\end{rem}


We will make use of the fact that the set of functions
\[
	\mathcal{C} \defn \left\{ F(\bmu) = \prod_{i = 0}^{M} \prod_{k = 1}^{K_{i}} \langle f_{ik}, \mu_{i} \rangle \middle\vert K_{i} \in \bbN_{0}, f_{ik} \in C([0,1]) \right\}
\]
 is separating, and convergence determining \cite{Ethier+Kurtz86}, so that for the purpose of characterizing our process and its limits, we need only compute the value the generator takes on functions $F \in \mathcal{C}$ and its limits.
 
We will evaluate the generator on this class of functions, but we first begin with a pair of lemmas.  We will use $\coprod$ to indicate the disjoint union of sets and, for all integers 
$M > 0$, we use the shorthand $[M] = \{1,\ldots,M\}$.

\begin{lem}\label{DECOMPOSE}
Let $\mu_i = \frac{1}{N_{i}} \sum_{j=1}^{N_{i}} \delta_{x_{ij}}$ for $x_{ij} \in [0,1]$.  Then,
\begin{multline*}
	\prod_{k = 1}^{K_{i}} \langle f_{ik}, \mu_i \rangle
	= \frac{1}{N_{i}^{K_{i}}} 
	\sum_{m=1}^{K_{i}} \sum_{{j_{1},\ldots,j_{m}} \atop{\text{distinct}}} \sum_{A_{1} \coprod \cdots \coprod A_{m} = [K_{i}]} \prod_{q = 1}^{m} \prod_{r \in A_{q}}  f_{ir}(x_{ij_{q}})\\
	= \frac{1}{N_{i}^{K_{i}}} \sum_{{j_{1},\ldots,j_{K_{i}}} \atop{\text{distinct}}} \prod_{k = 1}^{K_{i}} f_{ik}(x_{ij_{k}}) + \BigO{N^{-1}},
\end{multline*}
where the sum is over all partitions of $[K_{i}]$ into $m$ disjoint sets.
\end{lem}

\begin{proof}
The first statement is simply a matter of collecting terms according to the number of distinct values $j_{k}$:
\begin{multline*}
	N_{i}^{K_{i}} \prod_{k = 1}^{K_{i}} \langle  f_{ik}, \mu_i \rangle = N_{i}^{K_{i}} \prod_{k = 1}^{K_{i}} \left(\frac{1}{N_{i}} \sum_{j=1}^{N_{i}} f_{ik}(x_{ij})\right)
	= \sum_{j_{1}=1}^{N_{i}} \cdots \sum_{j_{K_{i}}=1}^{N_{i}} \prod_{k = 1}^{K_{i}} f_{ik}(x_{ij_{k}})\\
	= \sum_{m=1}^{K_{i}} \sum_{{j_{1},\ldots,j_{m}} \atop{\text{distinct}}} 
		\sum_{A_{1} \coprod \cdots \coprod A_{m} = [K_{i}]} \prod_{q = 1}^{m} \prod_{r \in A_{q}}  f_{ir}(x_{ij_{q}}).
\end{multline*}
Now, for the final term, we have $m = 1$, $A_{1} = [K_{i}]$, so it takes the form:
\[	
	 \sum_{j_{1}=1}^{N_{i}} \prod_{k = 1}^{K_{i}} f_{ik}(x_{ij_{1}})  = N_{i} \langle \prod_{k = 1}^{K_{i}}  f_{ik}, \mu_i \rangle,
\]
whilst for $m = 2$, we have:
\begin{multline*}
	= \sum_{j_{1}=1}^{N_{i}} \sum_{j_{2} \neq j_{1}}  \sum_{A_{1} \coprod A_{2} = [K_{i}]} \prod_{k \in A_{1}} f_{ik}(x_{ij_{1}}) \prod_{k \in A_{2}} f_{ik}(x_{ij_{2}}) \\
	= \sum_{A_{1} \coprod A_{2} = [K_{i}]} \sum_{j_{1}=1}^{N_{i}} \prod_{k \in A_{1}} f_{ik}(x_{ij_{1}}) \sum_{j_{2}=1}^{N_{i}} \prod_{k \in A_{2}} f_{ik}(x_{ij_{2}}) 
	 - \sum_{A_{1} \coprod A_{2} = [K_{i}]}  \sum_{j_{1}=1}^{N_{i}} \prod_{k \in A_{1}} f_{ik}(x_{ij_{1}}) \prod_{k \in A_{2}} f_{ik}(x_{ij_{1}}) \\	
	 = N_{i}^{2}  \sum_{A_{1} \coprod A_{2} = [K_{i}]} \langle \prod_{k \in A_{1}} f_{ik}, \mu_i \rangle \langle \prod_{k \in A_{2}} f_{ik}, \mu_i \rangle
	 - S(K_{i},2) N_{i} \langle \prod_{k \in A_{1}} f_{ik}, \mu_i \rangle, 
\end{multline*}
where $S(K_{i},2)$ is a Stirling number of the second kind \cite{Abramowitz+Stegun64} and gives the number of distinct partitions of $K_{i}$ elements into 2 sets.

Proceeding inductively in this manner completes the proof of the lemma.
\end{proof}

The previous lemma shows we will be interested in products over distinct indices $j_{1},\ldots,j_{m}$.  In particular, we have the result of Lemma~\ref{EXPPROD}.

\begin{lem}\label{EXPPROD}
For distinct values $j_{1},\dots,j_{K_{i}}$ in $\{1,\ldots,N_{i}\}$, 
\begin{multline*}
 	\bbE\left[ \prod_{k = 1}^{K_{i}} f_{ik}(X_{ij_{k}}(n+1)) \middle\vert \{X_{ij}(n) = x_{ij} \} \right]  = 
	\frac{\bbE\left[\nu_{i1}\cdots\nu_{iK_{i}}\right]}{(N_{i})_{K_{i}}} \sum_{{p_{1},\ldots,p_{K_{i}}} \atop{\text{distinct}}}  \prod_{k = 1}^{K_{i}} (Q^{(N)}_{i} f_{ik})(x_{ip_{k}})\\
	+ \frac{\bbE\left[(\nu_{i1})_{2}\nu_{i2}\cdots\nu_{iK_{i}-1}\right]}{(N_{i})_{K_{i}}} \sum_{q < r} \sum_{{p_{1},\ldots,p_{K_{i}}} \atop {p_{q}=p_{r}}}
		\prod_{{k = 1} \atop {k \neq q,r}}^{K_{i}} (Q^{(N)}_{i} f_{ik})(x_{ip_{k}}) (Q^{(N)}_{i} f_{iq} Q^{(N)}_{i} f_{ir})(x_{ip_{q}}) + o(c_{N_{i}}).
\end{multline*}
\end{lem}
 
\begin{proof}
We begin by recalling that conditional on an individual's parent having type $x$, its type is independently distributed according to the probability measure $P(x,\cdot)$, \ie
\[
	\bbE\left[ f(X_{ij}(n+1)) \middle\vert \bG^{(N)}(n), \text{parent of type $x$}\right] = (Q^{(N)}_{i}f)(x).
\]
We can thus, similar to the previous lemma, write:
\begin{multline*}
	\bbE\left[ \prod_{k = 1}^{K_{i}} f_{ik}(X_{ij_{k}}(n+1)) \middle\vert \{X_{ij}(n) = x_{ij} \} \right] \\
	= \sum_{m=1}^{K_{i}} \sum_{{p_{1},\ldots,p_{m}} \atop{\text{distinct}}} \sum_{A_{1} \coprod \cdots \coprod A_{m} = [K_{i}]} 
		\frac{\bbE\left[(\nu_{ip_{1}})_{\abs{A_{1}}}\cdots(\nu_{ip_{m}})_{\abs{A_{m}}}\right]}{(N_{i})_{K_{i}}} \prod_{q = 1}^{m} \prod_{r \in A_{q}}  (Q^{(N)}_{i} f_{ir})(x_{ip_{q}}),
\end{multline*}
where
\[
	\frac{\bbE\left[(\nu_{ip_{1}})_{\abs{A_{1}}}\cdots(\nu_{ip_{m}})_{\abs{A_{m}}}\right]}{(N_{i})_{K_{i}}} 
	= \frac{\bbE\left[(\nu_{i1})_{\abs{A_{1}}}\cdots(\nu_{im})_{\abs{A_{m}}}\right]}{(N_{i})_{K_{i}}}
\]
is the probability that the $K_{i}$ distinct individuals have $m$ ancestors $p_{1},\ldots,p_{m}$ (with types $x_{ip_{1}},\ldots,x_{p_{m}}$), and that the individuals in $A_{q}$ had parent $p_{q}$.

Next, we observe that since $\norm{Q^{(N)}_{i}} \leq 1$,
\begin{multline*}
	\abs{\sum_{{p_{1},\ldots,p_{m}} \atop{\text{distinct}}} \sum_{A_{1} \coprod \cdots \coprod A_{m} = [K_{i}]} 
		\frac{\bbE\left[(\nu_{ip_{1}})_{\abs{A_{1}}}\cdots(\nu_{ip_{m}})_{\abs{A_{m}}}\right]}{(N_{i})_{K_{i}}} \prod_{q = 1}^{m} \prod_{r \in A_{q}}  (Q^{(N)}_{i} f_{ir})(x_{ip_{q}})}\\	
	\leq \sum_{A_{1} \coprod \cdots \coprod A_{m} = [K_{i}]} \frac{(N_{i})_{m}}{(N_{i})_{K_{i}}} \bbE\left[(\nu_{i1})_{\abs{A_{1}}}\cdots(\nu_{im})_{\abs{A_{m}}}\right] 
		\prod_{k=1}^{K_{i}} \norm{f_{ik}}\\
\end{multline*}
and is thus $o(c_{N_{i}})$ whenever $\abs{A_{q}} \geq 3$ for some $q$ or $\abs{A_{q}}$ and $\abs{A_{r}}$ are both $\geq 2$ for distinct indices $q, r$ by \eqref{MOHLECONS2}.  The result follows.
\end{proof}

We now turn to the main result of this section:

\begin{prop}\label{INTERMEDIATEPROP}
Let $\mu^{(N)}_{i} = \frac{1}{N_{i}} \sum_{j=1}^{N_{i}} \delta_{x_{ij}}$, for $x_{ij} \in [0,1]$ and let $\bmu^{(N)} = \mu^{(N)}_{1} \otimes \cdots \otimes \mu^{(N)}_{M}$ converge weakly to a measure $\bmu \in \mathscr{P}([0,1])^{M+1}$.  

Let $F(\bmu) = \prod_{i = 0}^{M} \prod_{k = 1}^{K_{i}} \langle f_{ik}, \mu_{i} \rangle \in \mathcal{C}$ and, for $i=1,\ldots,M$, let
\begin{multline}\label{INTERGEN}
	(\mathcal{G}_{i}F)(\bmu) = \prod_{{j=0} \atop {j \neq i}}^{M} \prod_{k = 1}^{K_{j}} \langle f_{jk}, \mu_{i} \rangle
	\left(\sum_{q = 1}^{K_{i}} \frac{\varpi_{i}}{2} \langle f_{iq}, \mu_{i} - \mu_{0}\rangle \prod_{{k=1} \atop {k \neq q}}^{K_{i}} \langle f_{ik}, \mu_{i} \rangle\right.\\
	\left. + \frac{1}{2} \sum_{q \neq r} \prod_{{k = 1} \atop {k \neq q,r}}^{K_{i}} \langle  f_{ik}, \mu_{i} \rangle
		\left(\langle  f_{iq} f_{ir}, \mu_{i} \rangle - \langle  f_{iq}, \mu_{i} \rangle\langle  f_{ir}, \mu_{i} \rangle\right)\right)
\end{multline}
define an operator on $C(\mathscr{P}([0,1])^{M+1})$.

Then, 	
\[
	\lim_{N \to \infty} a_{N}^{-1} (\mathcal{G}_{N} F)(\bmu^{(N)}) = (\mathcal{G}F)(\bmu) \defn \sum_{i=1}^{M} \gamma_{i} (\mathcal{G}_{i} F)(\bmu).
\]
Moreover, given $\tilde{\mu}_{i} \in \mathscr{P}(\mathscr{P}([0,1]))$, there exist unique independent processes $G_{i}(t)$ with generators $\mathcal{G}_{i}$, such that $G_{i}(0)$ is distributed according to $\tilde{\mu}_{i}$ and such that 
\[
	\bG^{(N)}(\lfloor a_{N}^{-1} t \rfloor) \Rightarrow \bG(t) \defn G_{0}(0) \otimes G_{1}(\gamma_{1} t) \otimes \cdots \otimes G_{M}(\gamma_{M} t),
\]
for all $t > 0$, where convergence is in the space of c\`adl\`ag functions endowed with the Skorokhod topology, $\bbD_{\mathscr{P}([0,1])^{M+1}}[0,\infty)$ (see \eg \cite{Ethier+Kurtz86}).  
\end{prop}

\begin{rem}
Because 
\[
	\lim_{N \to \infty} \frac{c_{N}}{a_{N}} =  0,
\]
the component of the generator acting on the mainland vanishes in the limit; if 
\[
	\mathcal{C}_{0} \defn \left\{F \in \mathcal{C} \middle\vert F(\bmu) = \prod_{k = 1}^{K_{0}} \langle f_{ik}, \mu_{0} \rangle\right\},
\]
then $\mathcal{G}_{i} F \equiv 0$ for all $F \mathcal{C}_{0}$ and thus the generator vanishes on this set.  Equivalently, the process $G_{0}(t) \equiv \mu_{0}$.
\end{rem}

\begin{rem} 
Recall from \eqnref{GAMMADEF} that the effective population size of the \Th{i} island is $c_{N_{i}} \sim \gamma_{i} a_{N}$; since we have rescaled time by $a_{N}$ rather than the individual effective population sizes, the factors $\gamma_{i}$ appear in the generator and in the components $G_{i}$.  These reflect the fact that the different effective population sizes on the different islands result in their population dynamics having different rates (\ie different expected inter-event times), which are given by the $\gamma_{i}$.
\end{rem}

\begin{rem}
This theorem tells us that on the intermediate time scale, the islands have essentially independent dynamics, coupled only by immigration from a mainland which remains unchanged on the intermediate timescale.  The generator of the dynamics on the island is identical to that in the infinite population limit for the infinitely many alleles model, with the rescaled migration rate, $\frac{\varpi_{i}}{2}$ taking the place of the rescaled mutation rate $\theta$, and the mainland density measure $\mu_{0}$ taking the place of Lebesgue measure.
\end{rem}

\begin{proof} 
Applying Lemmas \ref{DECOMPOSE} and \ref{EXPPROD}, we have
\begin{multline*}
 	\mathbb{E}\left[ F(\bG^{(N)}(n+1)) \middle\vert \bG^{(N)}(n) = \bmu \right]
		= \mathbb{E}\left[ \prod_{i = 0}^{M} \prod_{k = 1}^{K_{i}} \langle f_{ik}, G^{(N)}_{i}(n+1) \rangle \middle\vert \{X_{ij}(n) = x_{ij}\} \right] \\
		= \prod_{i = 0}^{M} \mathbb{E}\left[\prod_{k = 1}^{K_{i}} \langle f_{ik}, G^{(N)}_{i}(n+1) \rangle \middle\vert \{X_{ij}(n)= x_{ij}\} \right] \\
		= \prod_{i = 0}^{M} \frac{1}{N_{i}^{K_{i}}} \sum_{m=1}^{K_{i}} \sum_{{j_{1},\ldots,j_{m}} \atop{\text{distinct}}} \sum_{A_{1} \coprod \cdots \coprod A_{m} = [K_{i}]} 
			 \mathbb{E}\left[\prod_{q = 1}^{m} \prod_{r \in A_{q}}  f_{ir}(x_{ij_{q}}) \middle\vert \{X_{ij}(n)= x_{ij}\} \right]\\
		= \prod_{i = 0}^{M} \frac{1}{N_{i}^{K_{i}}} \sum_{m=1}^{K_{i}} \sum_{{j_{1},\ldots,j_{m}} \atop{\text{distinct}}} \sum_{A_{1} \coprod \cdots \coprod A_{m} = [K_{i}]} 
		\left(\frac{\bbE\left[\nu_{i1}\cdots\nu_{im}\right]}{(N_{i})_{m}} \sum_{{p_{1},\ldots,p_{m}} \atop{\text{distinct}}}  
			\prod_{k = 1}^{m} (Q^{(N)}_{i}\prod_{l \in A_{k}}  f_{il})(x_{ip_{k}}) \right.\\
		+\left. \frac{\bbE\left[(\nu_{i1})_{2}\nu_{i2}\cdots\nu_{im-1}\right]}{(N_{i})_{m}} \sum_{q < r} \sum_{{p_{1},\ldots,p_{m}} \atop {p_{q}=p_{r}}}
		\prod_{{k = 1} \atop {k \neq q,r}}^{m} (Q^{(N)}_{i}\prod_{l \in A_{k}}  f_{il})(x_{ip_{k}}) ((Q^{(N)}_{i} \prod_{l \in A_{q}}  f_{il})(Q^{(N)}_{i} \prod_{l \in A_{r}}  f_{il}))(x_{ip_{q}}) 
		+ o(c_{N_{i}})\right).
\end{multline*}
Now, observing that the term in brackets is independent of the values $j_{k}$, we note that $j_{1},\ldots,j_{m}$ can be chosen in $(N_{i})_{m}$ ways, and we are left with a product over sums of the form:
\begin{multline*}
	\frac{1}{N_{i}^{K_{i}}} \sum_{m=1}^{K_{i}} \bbE\left[\nu_{i1}\cdots\nu_{im}\right] \sum_{A_{1} \coprod \cdots \coprod A_{m} = [K_{i}]} 
		\sum_{{p_{1},\ldots,p_{m}} \atop{\text{distinct}}} \prod_{k = 1}^{m} (Q^{(N)}_{i}\prod_{l \in A_{k}}  f_{il})(x_{ip_{k}})\\
		+ \frac{1}{N_{i}^{K_{i}}} \sum_{m=1}^{K_{i}} \bbE\left[(\nu_{i1})_{2}\nu_{i2}\cdots\nu_{im-1}\right] \sum_{A_{1} \coprod \cdots \coprod A_{m} = [K_{i}]} 
		\sum_{q < r} \sum_{{p_{1},\ldots,p_{m}} \atop {p_{q}=p_{r}}}\\
		\prod_{{k = 1} \atop {k \neq q,r}}^{m} (Q^{(N)}_{i}\prod_{l \in A_{k}}  f_{il})(x_{ip_{k}}) ((Q^{(N)}_{i} \prod_{l \in A_{q}}  f_{il})(Q^{(N)}_{i} \prod_{l \in A_{r}}  f_{il}))(x_{ip_{q}}) 
		+ o(c_{N_{i}}).
\end{multline*}

We will focus our attention on the first sum in the first line.  Using Lemma \ref{DECOMPOSE} in reverse, we have 
\begin{multline*}
	\frac{1}{N_{i}^{K_{i}}} \sum_{{p_{1},\ldots,p_{K_{i}}} \atop{\text{distinct}}} \prod_{k = 1}^{K_{i}} (Q^{(N)}_{i} f_{ik})(x_{ip_{k}}) = 	
	\prod_{k = 1}^{K_{i}} \langle Q^{(N)}_{i} f_{ik}, \mu_{i} \rangle\\
	- \frac{1}{N_{i}^{K_{i}}} \sum_{m=1}^{K_{i}-1} 
		\sum_{A_{1} \coprod \cdots \coprod A_{m} = [K_{i}]} \sum_{{p_{1},\ldots,p_{m}} \atop{\text{distinct}}}  \prod_{k = 1}^{m} (Q^{(N)}_{i}\prod_{l \in A_{k}}  f_{il})(x_{ip_{k}}),
\end{multline*}
where the terms on the second line are $\BigO{\frac{1}{N_{i}}}$.  Thus, 
\begin{multline*}
	\frac{1}{N_{i}^{K_{i}}} \sum_{m=1}^{K_{i}} \bbE\left[\nu_{i1}\cdots\nu_{im}\right] \sum_{A_{1} \coprod \cdots \coprod A_{m} = [K_{i}]} 
		\sum_{{p_{1},\ldots,p_{m}} \atop{\text{distinct}}} \prod_{k = 1}^{m} (Q^{(N)}_{i}\prod_{l \in A_{k}}  f_{il})(x_{ip_{k}})\\
	= \bbE\left[\nu_{i1}\cdots\nu_{iK_{i}}\right] \prod_{k = 1}^{K_{i}} \langle Q^{(N)}_{i} f_{ik}, \mu_{i} \rangle\\
	+ \frac{1}{N_{i}^{K_{i}}} \sum_{m=1}^{K_{i}-1} \left(\bbE\left[\nu_{i1}\cdots\nu_{im}\right] -  \bbE\left[\nu_{i1}\cdots\nu_{iK_{i}}\right]\right)
		\sum_{A_{1} \coprod \cdots \coprod A_{m} = [K_{i}]} \sum_{{p_{1},\ldots,p_{m}} \atop{\text{distinct}}}  \prod_{k = 1}^{m} (Q^{(N)}_{i}\prod_{l \in A_{k}}  f_{il})(x_{ip_{k}}).
\end{multline*}
Further, we observed in Remark \ref{DIFFERENCE} that the differences $\bbE\left[\nu_{i1}\cdots\nu_{im}\right] -  \bbE\left[\nu_{i1}\cdots\nu_{iK_{i}}\right]$ are $\BigO{c_{N_{i}}}$, 
so that the first sum reduces to 
\[
	\bbE\left[\nu_{i1}\cdots\nu_{iK_{i}}\right] \prod_{k = 1}^{K_{i}} \langle Q^{(N)}_{i} f_{ik}, \mu_{i} \rangle + o(c_{N_{i}}).
\]
Proceeding similarly, applying Lemma \ref{DECOMPOSE} with the set of $K_{i}-1$ distinct functions $\{Q^{(N)}_{i} f_{ik}\}_{k \neq q, r} \cup \{(Q^{(N)}_{i} f_{iq})(Q^{(N)}_{i} f_{ir})\}$, we see
that 
\begin{multline*}
	 \frac{1}{N_{i}^{K_{i}}} \sum_{m=1}^{K_{i}} \bbE\left[(\nu_{i1})_{2}\nu_{i2}\cdots\nu_{im-1}\right] \sum_{A_{1} \coprod \cdots \coprod A_{m} = [K_{i}]} 
		\sum_{q < r} \sum_{{p_{1},\ldots,p_{m}} \atop {p_{q}=p_{r}}}\\
		\prod_{{k = 1} \atop {k \neq q,r}}^{m} (Q^{(N)}_{i}\prod_{l \in A_{k}}  f_{il})(x_{ip_{k}}) ((Q^{(N)}_{i} \prod_{l \in A_{q}}  f_{il})(Q^{(N)}_{i} \prod_{l \in A_{r}}  f_{il}))(x_{ip_{q}}) \\
		= \frac{1}{N_{i}} \bbE\left[(\nu_{i1})_{2}\nu_{i2}\cdots\nu_{iK_{i}-1}\right] 
			\sum_{q < r} \prod_{{k = 1} \atop {k \neq q,r}}^{K_{i}} \langle Q^{(N)}_{i} f_{ik}, \mu_{i} \rangle \langle (Q^{(N)}_{i} f_{iq})(Q^{(N)}_{i} f_{ir}), \mu_i \rangle + o(c_{N_{i}}),
\end{multline*}
where we have used the fact that $\frac{1}{N_{i}} \bbE\left[(\nu_{i1})_{2}\nu_{i2}\cdots\nu_{im-1}\right] = \BigO{c_{N_{i}}}$ in bounding the higher order terms.  Thus,
\begin{multline*}
	\mathbb{E}\left[ F(\bG^{(N)}(n+1)) \middle\vert \bG^{(N)}(n) = \bmu \right]
	= \prod_{i = 0}^{M} \left(\bbE\left[\nu_{i1}\cdots\nu_{iK_{i}}\right] \prod_{k = 1}^{K_{i}} \langle Q^{(N)}_{i} f_{ik}, \mu_{i} \rangle\right.\\
		+ \left.\frac{1}{N_{i}} \bbE\left[(\nu_{i1})_{2}\nu_{i2}\cdots\nu_{iK_{i}-1}\right] 
			\sum_{q < r} \prod_{{k = 1} \atop {k \neq q,r}}^{K_{i}} \langle Q^{(N)}_{i} f_{ik}, \mu_{i} \rangle \langle (Q^{(N)}_{i} f_{iq})(Q^{(N)}_{i} f_{ir}), \mu_i \rangle + o(c_{N_{i}})		\right)\\
	= \prod_{i = 0}^{M} \bbE\left[\nu_{i1}\cdots\nu_{iK_{i}}\right] \prod_{k = 1}^{K_{i}} \langle Q^{(N)}_{i} f_{ik}, \mu_{i} \rangle
		+ \sum_{i = 0}^{M} \prod_{{j=0} \atop {j \neq i}}^{M} \bbE\left[\nu_{j1}\cdots\nu_{jK_{j}}\right]  \prod_{k = 1}^{K_{j}} \langle Q^{(N)}_{i} f_{jk}, \mu_{i} \rangle \\
			\times \frac{1}{N_{i}} \bbE\left[(\nu_{i1})_{2}\nu_{i2}\cdots\nu_{iK_{i}-1}\right]  
			\sum_{q < r} \prod_{{k = 1} \atop {k \neq q,r}}^{K_{i}} \langle Q^{(N)}_{i} f_{ik}, \mu_{i} \rangle \langle (Q^{(N)}_{i} f_{iq})(Q^{(N)}_{i} f_{ir}), \mu_{i} \rangle
			 + o(c_{N_{i}}).
\end{multline*}
Further, recalling \eqref{QI}, by assumption
\[
	Q^{(N)}_{i} = I + c_{N_{i}} B^{(N)}_{i} + \BigO{c_{N_{i}}},
\]
we have
\begin{multline}\label{FIRST}
	\mathbb{E}\left[ F(\bG^{(N)}(n+1)) \middle\vert \bG^{(N)}(n) = \bmu \right]\\
	=  \prod_{i = 0}^{M} \bbE\left[\nu_{i1}\cdots\nu_{iK_{i}}\right] \prod_{k = 1}^{K_{i}} \langle f_{ik}, \mu_i \rangle\\
	+ \sum_{i = 0}^{M} c_{N_{i}} \bbE\left[\nu_{i1}\cdots\nu_{iK_{i}}\right] \prod_{{j=0} \atop {j \neq i}}^{M} \bbE\left[\nu_{j1}\cdots\nu_{jK_{j}}\right] 
		\prod_{k = 1}^{K_{j}} \langle f_{jk}, \mu_i \rangle
		 \sum_{q = 1}^{K_{i}} \langle B^{(N)}_{i} f_{iq}, \mu_i \rangle \prod_{{k=1} \atop {k \neq q}}^{K_{i}} \langle f_{ik}, \mu_i \rangle \\
		+ \sum_{i = 0}^{M} \prod_{{j=0} \atop {j \neq i}}^{M} \bbE\left[\nu_{j1}\cdots\nu_{jK_{j}}\right]  \prod_{k = 1}^{K_{j}} \langle f_{jk}, \mu_i \rangle
			 \frac{1}{N_{i}} \bbE\left[(\nu_{i1})_{2}\nu_{i2}\cdots\nu_{iK_{i}-1}\right]  
			\sum_{q < r} \prod_{{k = 1} \atop {k \neq q,r}}^{K_{i}} \langle  f_{ik}, \mu_i \rangle \langle  f_{iq} f_{ir}, \mu_i \rangle \\
			+ o(c_{N_{i}}).
\end{multline}
Now recall,  
\[
	(B^{(N)}_{i}f)(x) \defn \frac{\varpi_{i}}{2} \left(\int f(y) G^{(N)}(n)(dy) - f(x)\right),
\]
and, from \eqref{METAPOP}, we have
\[
	G^{(N)}(n) = \frac{1}{\sum_{k=0}^{M} N_{k}} \sum_{i=0}^{M} N_{i} G^{(N)}_{i}(n) = \frac{1}{\sum_{k=0}^{M} N_{k}} \sum_{i=0}^{M} N_{i} \mu_{i} 
	= \mu_{0} + \BigO{\frac{N_{i}}{N_{0}}},
\]
so that
\[	
	(B^{(N)}_{i}f)(x) = \frac{\varpi_{i}}{2} \left(\int f(y) \mu_{0}(dy) - f(x)\right) + o(1).
\]

Now, recalling Lemma \ref{ONE}, we have
\begin{multline*}
	\bbE\left[\nu_{i1}\cdots\nu_{iK_{i}}\right] = 1- {K_{i} \choose 2} \frac{(N_{i})_{K_{i}-1}}{(N_{i})_{K_{i}}} \bbE\left[(\nu_{i1})_{2}\nu_{i2}\cdots\nu_{iK_{i}-1}\right] - o(c_{N_{i}})\\
		= 1- {K_{i} \choose 2} \frac{1}{N_{i}-K_{i}+1} \bbE\left[(\nu_{i1})_{2}\nu_{i2}\cdots\nu_{iK_{i}-1}\right] - o(c_{N_{i}})\\
		= 1- {K_{i} \choose 2} \left(\frac{1}{N_{i}}+\frac{K_{i}-1}{N_{i}(N_{i}-K_{i}+1)}\right) \bbE\left[(\nu_{i1})_{2}\nu_{i2}\cdots\nu_{iK_{i}-1}\right] - o(c_{N_{i}})\\
		= 1- {K_{i} \choose 2} \frac{1}{N_{i}} \bbE\left[(\nu_{i1})_{2}\nu_{i2}\cdots\nu_{iK_{i}-1}\right] - o(c_{N_{i}}),
\end{multline*}
so that 
\begin{multline}\label{SECOND}
	F(\bmu) = \prod_{i = 0}^{M} \prod_{k = 1}^{K_{i}} \langle f_{ik}, \mu_{i} \rangle\\
	= \prod_{i = 0}^{M} \left(\bbE\left[\nu_{i1}\cdots\nu_{iK_{i}}\right] \prod_{k = 1}^{K_{i}} \langle f_{ik}, \mu_{i} \rangle
	 + {K_{i} \choose 2} \frac{1}{N_{i}} \bbE\left[(\nu_{i1})_{2}\nu_{i2}\cdots\nu_{iK_{i}-1}\right] \prod_{k = 1}^{K_{i}} \langle f_{ik}, \mu_{i} \rangle + o(c_{N_{i}})\right)\\
	 = \prod_{i = 0}^{M} \bbE\left[\nu_{i1}\cdots\nu_{iK_{i}}\right] \prod_{k = 1}^{K_{i}} \langle f_{ik}, \mu_{i} \rangle\\
	 	+ \sum_{i = 0}^{M} \prod_{{j=0} \atop {j \neq i}}^{M} \bbE\left[\nu_{j1}\cdots\nu_{jK_{j}}\right]  \prod_{k = 1}^{K_{j}} \langle f_{jk}, \mu_j \rangle
			\frac{1}{N_{i}} \bbE\left[(\nu_{i1})_{2}\nu_{i2}\cdots\nu_{iK_{i}-1}\right] 
	 			\sum_{q < r}  \prod_{k = 1}^{K_{i}} \langle f_{ik}, \mu_{i} \rangle +  o(c_{N_{i}}).
\end{multline}

Thus, taking the difference of \eqref{FIRST} and \eqref{SECOND} and using Lemmas \ref{ONE} and \ref{TWO} respectively to replace $\bbE\left[\nu_{i1}\cdots\nu_{iK_{i}}\right]$
and $\frac{1}{N_{i}} \bbE\left[(\nu_{i1})_{2}\nu_{i2}\cdots\nu_{iK_{i}-1}\right]$ by $1-\BigO{c_{N_{i}}}$ and $c_{N_{i}} + o(c_{N_{i}})$, we see that 
\[
	\mathbb{E}\left[ F(\bG^{(N)}(n+1)) \middle\vert \bG^{(N)}(n) = \bmu \right] - F(\bmu) = 
		c_{N_{i}} \sum_{i = 0}^{M} (\mathcal{G}_{i} F)(\bmu) + o(c_{N_{i}}).
\]
The first assertion follows directly.

We now observe that, restricted to the space of functions
\[
	\mathcal{C}_{i} \defn \left\{ F(\bmu) = \prod_{k = 1}^{K_{i}} \langle f_{ij}, \mu_{i} \rangle \middle\vert K_{i} \in \bbN_{0}, f_{ij} \in C([0,1]) \right\} \subseteq 
		C(\mathscr{P}([0,1])),
\]
$\mathcal{G}_{i}$ is exactly the generator (4.4) of the infinitely many  alleles model of Chapter 10 of \cite{Ethier+Kurtz86}.  In particular, Theorem 4.1 of the same chapter tells us that given a fixed initial measure $\tilde{\mu}_{i} \in \mathscr{P}(\mathscr{P}([0,1]))$, the martingale problem for $(\mathcal{G}_{i},\tilde{\mu}_{i})$ is well posed, \ie there exists a unique in distribution process $G_{i}(t)$ with initial value $G_{i}(0)$ distributed according to $\tilde{\mu}_{i}$ with generator $\mathcal{G}_{i}$.   Moreover, using Theorem 1.1 of Chapter 6 of \cite{Ethier+Kurtz86}, we see that $G_{i}(\gamma_{i} t)$ is the unique process with generator $\gamma_{i} \mathcal{G}_{i}$.  We can thus appeal to Theorem 10.1 in  \cite{Ethier+Kurtz86} to conclude that given an initial measure $\tilde{\bmu} = \tilde{\mu}_{0} \otimes \cdots \otimes \tilde{\mu}_{M}$, then the martingale problem for
\[
	\sum_{i=1}^{M} \gamma_{i} \mathcal{G}_{i} 
\]
is well posed and has solution $G_{0}(0) \otimes G_{1}(\gamma_{1} t) \otimes \cdots \otimes G_{M}(\gamma_{M} t)$.

Given convergence of the generators, and well-posedness of the limiting generator, the second assertion then follows by Lemma 5.1 in Chapter 4 of \cite{Ethier+Kurtz86}.
\end{proof}

Finally, we conclude this section by observing that our characterization of the limiting generator in terms of the generator of the infinitely many alleles diffusion model also allows us to characterize the stationary distribution:

\begin{cor}\label{DPCOR}
The stationary process for the islands is the joint law of $M$ independent Dirichlet processes with scaling parameters $\varpi_{i}$ and base probability measure $\mu_{0}$,  
$\text{DP}(\varpi_{i},\mu_{0})$.
\end{cor}

\begin{proof}
This is immediate from the result for a single copy of the infinitely many alleles model.  See \eg Theorem 4.1, Chapter 9 in \cite{Ethier+Kurtz86}.
\end{proof}

\subsection{Long-Term Behaviour}\label{SLOW}

In the previous section, we simply assumed that the mainlands were asymptotically smaller in size (as measured by the coalescence probability of two randomly selected individuals) than the mainland, in order to show that the Cannings' UNTB converged to a sum of independent copies of the infinitely many alleles diffusion process, with migration from the mainland playing the role of mutation.  In this section, we will show that in a slow timescale, the dynamics on the mainland converge to the standard infinitely many alleles model as well, from which we can conclude, as before, that the stationary distribution of the mainland is that of the Dirichlet process $\text{DP}(\theta,\lambda)$, where $\lambda$ is the Lebesgue measure on $[0,1]$.  Thus, after a transient period, the mainland will approach a measure $\mu_{0} \sim \text{DP}(\theta,\lambda)$, whereas the islands will converge on Hierarchical Dirichlet Processes $\text{DP}(\varpi_{i},\mu_{0})$ \cite{Teh2006}.


Let $\tilde{\nu}_{i}$, $i=1,..,n$, be the law of the stationary process $\text{DP}(\varpi_{i},\mu_{0})$ from Corollary \ref{DPCOR} above, and let 
$\tilde{\nu} = \tilde{\nu}_{1} \otimes \cdots \otimes \tilde{\nu}_{M}$, \ie given a function $F \in C(\mathscr{P}([0,1])^{M})$, 
\[
	\int F(\bmu)\, \tilde{\nu}(d\bmu) = \int \cdots \int  F(\mu_{1},\ldots,\mu_{M})\, \tilde{\nu}_{1}(d\mu_{1}) \cdots \tilde{\nu}_{M}(d\mu_{M}),
\]
then $\tilde{\nu}$ is a stationary distribution for $\bG(t)$: we have 
\[
	\int (\cG F)(\bmu) \tilde{\nu}(d\bmu) = 0
\]
for all $F \in C(\mathscr{P}([0,1])^{M+1})$, or equivalently, writing $\cT(t)$ for the semi-group generated by $\cG$, (\ie
\[
	(\cT(t)F)(\bmu) = \bbE\left[F(\bG(t)) \middle\vert \bG(0)=\bmu\right],
\]
where $\bG(t)$ is the process with generator $\cG$ of Proposition \ref{INTERMEDIATEPROP}) we have 
\[
	\int (\cT(t) F)(\bmu) \tilde{\nu}(d\bmu) =  \int F(\bmu)\, \tilde{\nu}(d\bmu)
\]
for all $F \in C(\mathscr{P}([0,1])^{M})$.

We start by showing that as $t \to \infty$, $\bG(t)$ converges to a stationary process $\bG^{\star}$ distributed according to $\tilde{\nu}$, (\ie
\[
	\bbP\left\{\bG^{\star}(t) \in A \middle\vert \bG^{\star}(0) \sim \tilde{\nu} \right\} = \tilde{\nu}(A)
\]
for all subsets $A \subseteq \mathscr{P}([0,1])^{M+1}$).  To this end, we begin with a series of lemmas, which are essentially the same as results appearing in \cite{Ethier1981}:

\begin{lem}
Let $\bmu = \mu_{0} \otimes \cdots \otimes \mu_{M} \in \mathscr{P}([0,1])^{M+1}$ and let $F(\bmu) = \prod_{i = 0}^{M} \prod_{k=1}^{K_{i}} \langle f_{ik}, \mu_{i}\rangle \in \mathcal{C}$.  Let $K = \sum_{i = 0}^{M} K_{i}$ be the \textit{degree} of $F$.  If $K \geq 1$, there exists a scalar $\lambda > 0$ and a function $\psi$, which is a sum of functions of the same form as $F$, but of degree $K-1$, such that 
\[
	\mathcal{G}F = -\lambda F + \psi.
\]
Thus, 
\[
	(\mathcal{T}(t)F)(\bmu) = e^{-\lambda t} F + \int_{0}^{t} e^{-\lambda (t-s)} \mathcal{T}(s)\psi\, ds.
\]
\end{lem}

\begin{proof}
Recalling Equation \eqref{INTERGEN}, we have
\begin{multline*}
	(\mathcal{G}_{i}F)(\bmu) 
		= \sum_{i = 0}^{M} \sum_{1 \leq j \neq k \leq K_{i}} \left(\langle f_{ij}f_{ik}, \mu_{i} \rangle -  \langle f_{ij}, \mu_{i} \rangle  \langle f_{ik}, \mu_{i} \rangle\right)
		 \prod_{l \neq j,k} \langle f_{il}, \mu_{i}\rangle \\
		 + \sum_{i = 0}^{M} \sum_{j=1}^{K_{i}} \frac{\varpi_{i}}{2} \langle f_{ij}, x_{0} - \mu_{i} \rangle \prod_{k \neq j}  \langle f_{ik}, \mu_{i}\rangle\\
	= -\left(\sum_{i = 0}^{M} \frac{K_{i}(K_{i}-1)}{2} + \frac{\varpi_{i}}{2}\right) F
	+ \sum_{i = 0}^{M} \langle \sum_{1 \leq j \neq k \leq K_{i}}  f_{ij}f_{ik}, \mu_{i} \rangle \prod_{l \neq j,k} \langle f_{il}, \mu_{i}\rangle\\
	+ \sum_{i = 0}^{M}  \langle \sum_{j=1}^{K_{i}} \frac{\varpi_{i}}{2} f_{ij}, x_{0} \rangle  \prod_{k \neq j}  \langle f_{ik}, \mu_{i}\rangle,
\end{multline*}
giving the first statement.  In particular, if $K=1$, say $K_{i}=1$, we have 
\[
	\mathcal{G}F = -\frac{\gamma_{i}\varpi_{i}}{2} F + \langle \frac{\gamma_{i}\varpi_{i}}{2} f_{i1}, x_{0} \rangle.
\]

For the second statement, we observe that
\[
	\frac{d}{dt} e^{\lambda t} \mathcal{T}(t)F = e^{\lambda t}\left(\lambda \mathcal{T}(t)F + \mathcal{T}(t)\mathcal{G}F\right) = e^{\lambda t} \mathcal{T}(t)\psi.
\]
The result follows by integrating both sides over $(0,t)$.
\end{proof}

With this lemma, we can show that the process $\bG(t)$ is ergodic, \ie the distribution of $\bG(t)$ converges on $\tilde{\nu}$, independently of the initial condition.

\begin{prop}
Let $F \in C(\mathscr{P}([0,1])^{M})$.  As $t \to \infty$, 
\[
	\lim_{t \to \infty} \norm{\mathcal{T}(t)F -   \int F(\bmu)\, \tilde{\nu}(d\bmu)} = 0.
\]
\end{prop}

\begin{proof}
Since they are convergence-determining, it suffices to show the result for functions of the form  $F \in \mathcal{C}$.  We then have
\[
	(\mathcal{T}(t)F) = e^{-\lambda t} F + \int_{0}^{t} e^{-\lambda (t-s)} \mathcal{T}(s)\psi\, ds
\]
for $\lambda > 0$ and $\psi$ of degree $K-1$.  Integrating both sides, and recalling that 
\[
	\int (\cT(t) F)(\bmu) \tilde{\nu}(d\bmu) =  \int F(\bmu)\, \tilde{\nu}(d\bmu)
\]
we have
\[
	 \int F(\bmu)\, \tilde{\nu}(d\bmu) = e^{-\lambda t} \int F(\bmu)\, \tilde{\nu}(d\bmu) + \int_{0}^{t} e^{-\lambda (t-s)} \mathcal{T}(s)\int \psi(\bx)\, \tilde{\nu}(d\bmu)\, ds,
\]
so that
\[
	\norm{\mathcal{T}(t)F - \int F(\bmu)\, \tilde{\nu}(d\bmu)} 
	\leq e^{-\lambda t} \int F(\bmu)\, \tilde{\nu}(d\bmu) + \int_{0}^{t} e^{-\lambda (t-s)} \norm{\mathcal{T}(s)\psi - \int \psi(\bmu)\, \tilde{\nu}(d\bmu)} \, ds.
\]
The first term on the right hand side clearly vanishes as $t \to \infty$; for the latter, we can iterate the above inequality, relying on the fact that the process will eventually terminate when the degree reaches 1; when $K = 1$, say $\psi(\bmu) = \langle f_{i1},\mu_{i}\rangle$, we have
\[
	(\mathcal{T}(t)\psi) = e^{-\frac{\omega_{i}}{2} t} \psi + \int_{0}^{t} e^{-\frac{\omega_{i}}{2} (t-s)} \mathcal{T}(s) \langle \frac{\omega_{i}}{2} f_{i1}, x_{0} \rangle\, ds
	= e^{-\frac{\omega_{i}}{2} t} \psi + \int_{0}^{t} e^{-\lambda (t-s)} \langle \frac{\omega_{i}}{2} f_{i1}, x_{0} \rangle\, ds
\]	
whereas
\[
	 \int \psi(\bmu)\, \tilde{\nu}(d\bmu) = e^{-\frac{\omega_{i}}{2} t} \int \psi(\bmu)\, \tilde{\nu}(d\bmu) + \int_{0}^{t} e^{-\lambda (t-s)} \langle \frac{\omega_{i}}{2} f_{i1}, x_{0} \rangle\, ds,
\]
so that
\[
	\norm{\mathcal{T}(t)\psi - \int \psi(\bmu)\, \tilde{\nu}(d\bmu)}  = e^{-\frac{\omega_{i}}{2} t} \norm{\psi - \int \psi(\bmu)\, \tilde{\nu}(d\bmu)} \to 0
\]
as $t \to \infty$.
\end{proof}

Define a linear map $\mathcal{P}$ on $C(\mathscr{P}([0,1])^{M})$ by 
\[
	\mathcal{P}F = \int F(\bmu)\, \tilde{\nu}(d\bmu), 
\]
\ie $\mathcal{P}$ sends $F \in C(\mathscr{P}([0,1])^{M})$ to a constant function; more generally, if $F \in C(\mathscr{P}([0,1])^{M+1})$,  $\mathcal{P}F$ is a function of $\mu_{0}$ alone.  In particular, 
\[
	(\mathcal{P}F)(\mu_{0}) = \bbE\left[F(\mu_{0},G_{1}(t),\ldots,G_{M}(t)) \middle\vert G_{i}(0) \sim \tilde{\nu}_{i}\right],
\]
so that applying the operator $\mathcal{P}$ is equivalent to conditioning on the islands being at their stationary state. 

Note that $\mathcal{P}^{2} = \mathcal{P}$, so that $\mathcal{P}$ is a projection.  Moreover,
\[
	\mathcal{P}(\mathcal{G}F) = \int (\mathcal{G}F)\, \tilde{\nu}(d\bmu) = 0,
\]
so the range of $\mathcal{G}$ is contained in the null space of $\mathcal{P}$, $\mathscr{R}(\mathcal{G}) \subseteq \mathscr{N}(\mathcal{P})$, whereas $\mathcal{G} 1 =  0$, so that $\mathscr{R}(\mathcal{P}) \subseteq \mathscr{N}(\mathcal{G})$.  In fact, we have:

\begin{lem}
$\mathcal{P}$ is the spectral projection onto $\mathscr{N}(\mathcal{G})$.
\end{lem}

\begin{proof}
By definition, the spectral projection onto $\mathscr{N}(\mathcal{G})$, $\mathcal{Q}$,  is the residue of the resolvent of $\mathcal{G}$ at $\lambda = 0$:
\[
	\mathcal{Q} = \lim_{\lambda \to 0^{+}} \lambda (\lambda - \mathcal{G})^{-1} = \lim_{\lambda \to 0^{+}} \lambda \int_{0}^{\infty} e^{-\lambda t} \mathcal{T}(t)\, dt.
\]
Now, fix $\varepsilon > 0$ and choose $t_{0} > 0$ so that  $\norm{\mathcal{T}(t) - \mathcal{P}} < \varepsilon$ for $t > t_0$.  Then, for $\lambda > 0$,
 \begin{align*}
 	\norm{\lambda \int_{0}^{\infty} e^{-\lambda t} \mathcal{T}(t) \, dt - \mathcal{P}} 
	&= \norm{\lambda \int_{0}^{\infty} e^{-\lambda t} \left(\mathcal{T}(t) - \mathcal{P}\right) \, dt}\\
	&\leq \lambda \int_{0}^{\infty} e^{-\lambda t} \norm{\mathcal{T}(t) - \mathcal{P}} \, dt\\
	&= \lambda \int_{0}^{t_{0}} e^{-\lambda t} \norm{\mathcal{T}(t) - \mathcal{P}} \, dt + \lambda \int_{t_{0}}^{\infty} e^{-\lambda t} \norm{\mathcal{T}(t) - \mathcal{P}} \, dt\\
	& \leq \lambda t_{0} \sup_{t \leq t_{0}} \norm{\mathcal{T}(t) - \mathcal{P}} + \varepsilon.
\end{align*}
$\norm{\mathcal{T}(t) - \mathcal{P}}$ is a continuous function, and thus bounded on $[0,t_{0}]$.   Thus the first term vanishes as $\lambda \to 0^{+}$, whereas $\varepsilon$ can be chosen arbitrarily small.  We conclude $\mathcal{Q} = \mathcal{P}$.  
\end{proof}

With this, we are able to obtain our final result.

\begin{prop}\label{SLOWPROP}
Assume, as before, that
\[
	\lim_{N \to \infty} \frac{c_{N}}{a_{N}} = 0.
\]	
Let $\mathcal{P}$ be the projection defined above.  Define an operator $\mathcal{G}_{0}$ on $\mathcal{C}_{0}$ by 
\begin{multline}\label{SLOWGEN}
	(\mathcal{G}_{0}F)(\bmu) = 
	\left(\sum_{q = 1}^{K_{0}} \frac{\theta}{2} \langle f_{0q}, \lambda - \mu_{0}\rangle \prod_{{k=1} \atop {k \neq q}}^{K_{0}} \langle f_{0k}, \mu_{0} \rangle\right.\\
	\left. + \frac{1}{2} \sum_{q \neq r} \prod_{{k = 1} \atop {k \neq q,r}}^{K_{0}} \langle  f_{0k}, \mu_{0} \rangle
		\left(\langle  f_{0q} f_{0r}, \mu_{0} \rangle - \langle  f_{0q}, \mu_{0} \rangle\langle  f_{0r}, \mu_{0} \rangle\right)\right),
\end{multline}
and let $\mathcal{T}_{0}(t)$ be the semigroup generated by $\mathcal{P}\mathcal{G}_{0}$.  Then, for all $F \in C(\mathscr{P}([0,1])^{M})$, and all $\delta \in (0,1)$ we have 
\[
	\left(I+\mathcal{G}^{(N)}\right)^{\lfloor c_{N}^{-1} t \rfloor}F \to \mathcal{T}_{0}(t)\mathcal{P}F
\]
uniformly in $\delta \leq t \leq \delta^{-1}$.  If in addition, we assume that $G_{i}(0) \sim \tilde{\nu}_{i}$ for all $i=1,\ldots,M$, and $G_{0}(t)$ is a stochastic process with generator $\mathcal{G}_{0}$, then 
\[
	\bG^{(N)}(\lfloor c_{N}^{-1} t \rfloor) \Rightarrow \bG(t) = G_{0}(t) \otimes G_{1}(t) \cdots \otimes G_{M}(t),
\]
where the processes $G_{i}(t)$ are stationary for all $i=1,\ldots,M$.
\end{prop}

\begin{rem}
The heuristic understanding of Proposition \ref{SLOWPROP} is that 
\[
	\cG^{(N)} = c_{N}^{-1} a_{N} \cG + c_{N}^{-1} \cG_{0} + \text{lower order terms}
\]
where $\mathcal{H}\mathcal{P} \equiv 0$.  Now $c_{N}^{-1} a_{N} \to \infty$ as $N \to \infty$, so the first term dominates.  $c_{N}^{-1} a_{N}$ is essentially the rate at which the first term shapes the dynamics of the process, and so as $N$ grows large, the first term, which acts only on the islands, causes them to rapidly approach their equilibrium state (which, as we have already seen, corresponds to projection by $\mathcal{P}$).  The first term, however, has no effect on the mainland.  Moreover, the mainland only changes at the slower rate $c_{N}^{-1}$.  Thus, the first term has already forced the faster terms to equilibrium, and we can assume that they are at equilibrium when we consider the mainland.  Finally, the first two terms completely specify the limit, so what remains can only contribute a higher order correction.  This is essentially the infinite dimensional analogue of the following simple dynamical system:
\begin{gather*}
	\dot{x} = -N a x + f(x,y),\\
	\dot{y} = -\sqrt{N} b y + g(x,y),
\end{gather*} 
for $a,b > 0$.  Using variation of constants, we have 
\begin{gather*}	
	x(t) = e^{-N a t} x(0) + \int_{0}^{t} e^{-N a(t-s)} f(x(s),y(s))\, ds,\\ 
	y(t) = e^{-\sqrt{N} b t} y(0) + \int_{0}^{t} e^{-\sqrt{N} b(t-s)} g(x(s),y(s))\, ds.
\end{gather*}
Thus, provided $f$ and $g$ are bounded, 
\[	
	\int_{0}^{t} e^{-N a(t-s)} f(x(s),y(s))\, ds \leq \frac{1}{N a} \norm{f},
\]
and
\[	
	\int_{0}^{t} e^{-\sqrt{N} b(t-s)} g(x(s),y(s))\, ds \leq \frac{1}{\sqrt{N} b} \norm{g},
\]
so that as $N \to \infty$, we have $x(t) = 0 + \BigO{\frac{1}{N}}$.  We can thus substitute this back into the equation for $y(t)$ to conclude that
\[
	y(t) = e^{-\sqrt{N} b t} y(0) + \int_{0}^{t} e^{-\sqrt{N} b(t-s)} g(0,y(s))\, ds + \BigO{\frac{1}{N}},
\]
(setting $x(t) \equiv 0$ is equivalent to the action of the projection $\mathcal{P}$).  Thus, similarly, $y(t) = 0 + \BigO{\frac{1}{\sqrt{N}}}.$
\end{rem}

\begin{rem}
It is necessary to assume $G_{i}(0) \sim \tilde{\nu}_{i}$ to obtain continuity of  $\mathcal{T}_{0}(t)\mathcal{P}$ at $t = 0$, which in turn is required to ensure weak convergence.  More generally, Proposition \ref{SLOWPROP} tell us that in the slow timescale, the island demes instantaneously jump to their stationary states, and henceforth evolve as stationary processes; see \cite{Katzenberger1991} and \cite{Parsons2012} for more detailed discussions of processes with this behaviour.
\end{rem}

\begin{proof}
Calculations essentially identical to those in Proposition \ref{INTERMEDIATEPROP} show that, when restricted to $\mathcal{C}_{0}$, $c_{N}^{-1} \cG^{(N)} = \cG_{0} + o(c_{N})$, with the primary difference being with the operator $Q^{(N)}_{0}$. Here, 
\[
	Q^{(N)}_{0} =  I + c_{N} B^{(N)}_{0} + c_{N} B + o(c_{N}),
\]
where, as before
\[
	(B^{(N)}_{0}f)(x) = \frac{\varpi_{i}}{2} \left(\langle f, \mu_{0}\rangle - f(x)\right) + o(1),
\]
but now
\[
	(B f)(x) = \frac{\theta}{2} \int_{0}^{1} f(y)\, dy - f(x) = \theta (\langle f, \lambda\rangle - f(x))
\]
(recall that $\lambda$ is Lebesgue measure, $\lambda(dx) = dx$) is of the same asymptotic order.  Moreover, we now only consider terms of the form $\langle Q^{(N)}_{0} f_{0k}, \mu_{0} \rangle$, and
\[	
	\langle B^{(N)}_{0} f_{0k}, \mu_{0} \rangle	 = \frac{\varpi_{i}}{2}  \left(\langle f, \mu_{0}\rangle - \langle f, \mu_{0}\rangle\right) + o(1),
\]
which vanishes in the limit.  Thus, 
\[
	c_{N} \langle Q^{(N)}_{0} f_{0k}, \mu_{0} \rangle -  \langle f_{0k}, \mu_{0} \rangle =  \frac{\theta}{2} \int_{0}^{1} f(y)\, dy - f(x) + o(1) = \theta (\langle f, \lambda\rangle - f(x)) + o(1),
\]
giving the corresponding terms in the generator \eqref{SLOWGEN}.

The first statement is then a restatement of Corollary 7.7, Chapter 1 of \cite{Ethier+Kurtz86}; translating our notation into theirs, we have
\begin{gather*}
	\varepsilon_{N} = c_{N},\\
	\alpha_{N} = c_{N}^{-1} a_{N},\\
	A_{N} = c_{N}^{-1} \cG^{N},
\end{gather*}
$B = \cG$, and $A = \cG_{0}.$  That $\cG_{0}$ generates a strongly continuous semigroup is Theorem 4.1, Chapter 10 of \cite{Ethier+Kurtz86}, which we used previously. 

The second statement is a consequence of Corollary 8.9, Chapter 4, \cite{Ethier+Kurtz86}, where our initial condition ensures continuity of the semigroup $\cT_{0}(t)$ at $t = 0$. 
\end{proof}

\section{Gibbs Sampling for the UNTB-HDP}

\subsection{Observed abundances}

The observed data takes the form of an $N \times S$ matrix of counts $\mathbf{X}$ whose elements $x_{ij}$ are the observed frequency of species $j$ in community sample $i$.  Here, $N$ denotes the total number of communities and $S$ the total number of different species found in those communities.  We will also denote the row vectors of $\mathbf{X}$, which give the observed frequency distribution of species in each individual sample, by $\bar{X}_i$, $i = 1, \dots, N$. The size of each sample is simply $J_i = \sum_{j=1}^S x_{ij}$.

\subsection{Neutral-HDP model}

\begin{align}
\bar{\beta} | \theta &\sim \mbox{Stick} (\theta),\\
\bar{\pi}_i | I_i,\bar{\beta} &\sim \mbox{DP} (I_i,\bar{\beta}),\\
\bar{X}_i | \bar{\pi}_i, J_i & \sim \mbox{MN} (J_i,\bar{\pi}_i).
\end{align}
This model for the observed frequencies can be interpreted as the generation of an infinite dimensional metacommunity distribution $\bar{\beta}$ which is obtained from a stick-breaking or GEM distribution with concentration parameter $\theta$. From this, for each community $i$ we sample using the Dirichlet process a vector of taxa probabilities $\bar{\pi}_i$ which has concentration $I_i$, the immigration rate for that site, and base distribution $\bar{\beta}$. Finally, we sample the observed frequencies for each community $\bar{X}_i$ from $\bar{\pi}_i$ using the multinomial distribution. We also include gamma hyper-priors for $\theta$ and the $I_i$: 
\begin{align}
\theta | \alpha,\zeta &\sim \mbox{Gamma} (\alpha,\zeta), \eqnlabel{theta_priorA} \\
I_i | \eta,\  &\sim \mbox{Gamma} (\eta,\kappa),
\end{align}
where $\alpha,\zeta, \eta \text{ and } \kappa$ are all constants. This completes the definition of our model.

\subsection{Finite dimensional representation}

In any given sample although the potential number of species is infinite we only observe $S$ different types. It is convenient therefore to represent the model in terms of these finite dimensional number of types and one further class corresponding to all unobserved species. We will derive this as the limit of $L$ total types as $L \rightarrow \infty$. We will represent the proportions of the $S$ observed species explicitly as $\beta_k$ with $k = 1,\ldots,S$ and the unrepresented component as $\beta_u = \sum_{k = S + 1}^L \beta_k$. Let $\theta_r = \theta/L$ and $\theta_u = \theta(L - S)/L$, then we will have a Dirichlet prior on $\bar{\beta} ~ \sim \mbox{Dir}(\theta_r,\ldots,\theta_r,\theta_u)$.
In this finite dimensional representation we can also determine the distributions in the local communities:
\begin{equation}
\bar{\pi}_i ~ \sim \mbox{Dir}(I_i\beta_1,\ldots,I_i\beta_S,I_i\beta_u).
\end{equation}
We can then marginalise the local community distributions and derive the probability of the observed frequencies given $\bar{\beta}$:
\begin{equation}
P(\mathbf{X}|\bar{\beta},I_1,\ldots,I_N) = \prod_{i = 1}^N \frac{J_i!}{X_{i1}! \cdots X_{iS}!}\frac{\Gamma(I_i)}{\Gamma(J_i + I_i)} \prod_{j = 1}^S \frac{\Gamma(x_{ij} + I_i\beta_j)}{\Gamma(I_i\beta_j)}.
\eqnlabel{marginallikelihoodA}
\end{equation} 

\subsection{Gibbs sambling}

To devise a Gibbs sampling strategy we need to determine the full conditional distributions of the parameters we wish to sample, $\theta$ and $I_i$, for $i = 1, \dots, N$. Our starting point will be the joint distribution of these parameters and the data, that is, \eqnref{marginallikelihoodA} multiplied by the prior distributions for $\bar{\beta}$, $\theta$ and $I_i$, marginalised over $\bar{\beta}$:
\begin{equation}
P(\theta,I_1,\ldots,I_N, \mathbf{X}) = \int_{\bar{\beta}} P(\mathbf{X}|\bar{\beta},I_1,\ldots,I_N) P(\bar{\beta}|\theta) d\bar{\beta} \mbox{Gamma}(\theta|\alpha,\zeta)\prod_{i=1}^N\mbox{Gamma}(I_i|\eta,\nu).
\eqnlabel{posterior}
\end{equation}
The key to simplifying this expression is to expand the terms $\Gamma(x_{ij} + I_i\beta_j)/\Gamma(I_i \beta_j)$ in \eqnref{marginallikelihood} as polynomials \cite{Teh2006}:
\begin{equation}
\frac{\Gamma(x_{ij} + I_i\beta_j)}{\Gamma(I_i\beta_j)} = \sum_{T_{ij} = 0}^{T_{ij} = x_{ij}} s(x_{ij},T_{ij})(I_i\beta_j)^{T_{ij}},
\eqnlabel{Stirling_polynomials}
\end{equation}
where the coefficients $s(x_{ij},T_{ij})$ are unsigned Stirling numbers of the first kind. We subsitute these sums into \eqnref{posterior} and then introduce the $T_{ij}$ and $\bar{\beta}$ as auxilliary variables to give:
\begin{multline}
Q(\theta,\bar{\beta},I_1,\ldots,I_N,T_{ij}) \propto \left( \prod_{i = 1}^N \frac{J_i!}{X_{i1}! \cdots X_{iS}!}\frac{\Gamma(I_i)}{\Gamma(J_i + I_i)} \prod_{j=1}^S s(x_{ij},T_{ij})(I_i\beta_j)^{T_{ij}} \right) \\
P(\bar{\beta}|\theta)\mbox{Gamma}(\theta|\alpha,\zeta)\prod_{i=1}^N\mbox{Gamma}(I_i|\eta,\nu).
\eqnlabel{augmented}
\end{multline}

\subsubsection{Full conditional for the ancestral states}
 
From~\eqnref{augmented}, we see that the full conditional distribution for the number of ancestors (tables in the Chinese restaurant franchise analogy) of species $j$ in sample $i$ is given by:
\begin{equation}
\eqnlabel{eqn:Tables1}
P(T_{ij} | x_{ij}, I_i, \beta_j) \propto s(x_{ij}, T_{ij}) (I_i \beta_j)^{T_{ij}}.
\end{equation}
The reciprocal of~\eqnref{Stirling_polynomials} is the normalising constant of this probability distribution and thus:
\begin{equation}
\eqnlabel{eqn:Tables2A}
P(T_{ij} | x_{ij}, I_i, \beta_j) = \frac{\Gamma(I_i\beta_j)}{\Gamma(x_{ij} + I_i\beta_j)} s(x_{ij}, T_{ij}) (I_i \beta_j)^{T_{ij}}.
\end{equation}

\subsubsection{Full conditional for the metapopulation}

In their derivation of a posterior sampling scheme for the hierarchical Dirichlet process mixture model using an augmented Chinese restaurant franchise representation, \cite{Teh2006} showed that the full conditional distribution for the metapopulation vector $\bar{\beta}$ was:
\begin{equation}
\bar{\beta} = (\beta_1, \beta_2, \dots, \beta_S, \beta_u) \sim \mbox{Dir} (T_{\cdot 1}, T_{\cdot 2}, \dots, T_{\cdot S}, \theta),
\eqnlabel{eqn:DirichletDistBeta}
\end{equation}
where $T_{\cdot j} = \sum_{i=1}^N T_{ij}$.

\subsubsection{Full conditional for the immigration rates}

To derive the full conditional distribution of each $I_i$ given the other parameters we simply pull out all terms that depend on $I_i$ from \eqnref{augmented}. This gives:
\begin{equation}
P(I_i | T_{ij}) \propto \frac{\Gamma(I_i)}{\Gamma(J_i + I_i)}I_i^{T_{i \cdot}}\mbox{Gamma}(I_i|\eta,\nu),
\end{equation}
where $T_{i \cdot} = \sum_{j=1}^S T_{ij}$.
We can use the auxiliary variable approach of~\cite{EscobarWest95} to develop a Gibbs sampling update for $I_i$, $i = 1, \dots, N$.
Here, for each $i$, we can write:
\begin{equation}
\frac{\Gamma(I_i)}{\Gamma(I_i + J_i)} = \frac{1}{\Gamma(J_i)} \int_0^1 w_i^{I_i} (1 - w_i)^{J_i - 1} \left( 1 + \frac{J_i}{I_i} \right) \mathrm{d} w_i
\end{equation}
(cf. with equation (A.2) of~\cite{Teh2006}).  We now define auxiliary variables $\bar{w} = (w_i)_{i=1}^N$ and $\bar{s} = (s_i)_{i=1}^N$, where each $w_i$ is a variable taking on values in $[0,1]$ and each $s_i$ is a binary $\{0,1\}$ variable, and define the following distribution:
\begin{equation}
q(I_i, \bar{w}, \bar{s}) \propto \prod_{i=1}^N I_i^{\eta - 1 + T_{i \cdot}} e^{- \nu I_i} w_i^{I_i} (1 - w_i)^{J_i - 1} \left( \frac{J_i}{I_i} \right) ^{s_i}
\end{equation}
(cf. with equation (A.3) of~\cite{Teh2006}).  Now marginalising $q$ to $I_i$ gives the desired conditional distribution for $I_i$.  Hence $q$ defines an auxiliary variable sampling scheme for $I_i$.  Given $\bar{w}$ and $\bar{s}$, we have:
\begin{equation}
q(I_i | \bar{w}, \bar{s}) \propto I_i^{\eta - 1 + T_{i \cdot} - s_i} e^{- I_i ( \nu - \log w_i)},
\end{equation}
which is a Gamma distribution with parameters $\eta + T_{i \cdot} - s_i$ and $\nu - \log w_i$ (cf. with equation (A.4) of~\cite{Teh2006}).  Given $I_i$, the $w_i$ and $s_i$ are conditionally independent, with distributions:
\begin{equation}
q(w_i | I_i) \propto w_i^{I_i} (1 - w_i)^{J_i - 1}
\end{equation}
and
\begin{equation}
q(s_i | I_i) \propto \left( \frac{J_i}{I_i} \right) ^{s_i},
\end{equation}
which are $\mbox{Beta} (I_i + 1, J_i)$ and $\mbox{Bernoulli} \left( \frac{J_i}{J_i + I_i} \right)$, respectively (cf. with equations (A.5) and (A.6) of~\cite{Teh2006}).

\subsubsection{Full conditional for the biodiversity parameter}
\label{sec:Gibbs_theta}

A direct consequence of the stick-breaking prior for $\bar{\beta}$ is that the probability of observing $S$ species from a total number of $T = \sum_{i=1}^N \sum_{j=1}^S T_{ij}$ ancestors is given by:
\begin{equation}
P(S | \theta, T) = s(T,S) \theta^S \frac{\Gamma(\theta)}{\Gamma(\theta + T)}
\eqnlabel{eqn:species_distA}
\end{equation}
(cf. with equation (A.7) of~\cite{Teh2006}). The biodiversity parameter $\theta$ does not govern any other aspects of the joint distribution in~\eqnref{augmented}, hence~\eqnref{eqn:species_dist}, along with the prior for $\theta$ in~\eqnref{theta_priorA}, is all that is needed to derive a Gibbs sampling update for $\theta$. The auxiliary variable approach of~\cite{EscobarWest95} can also be applied here, which leads to the following auxiliary variable sampling scheme for $\theta$:
\begin{align}
\theta | \rho, \phi, S &\sim \mathrm{Gamma}( \alpha + S - \rho, \zeta - \log \phi), \\
\rho | \theta, T &\sim \mathrm{Bernoulli} \left( \frac{T}{T + \theta} \right), \\
\phi | \theta, T &\sim \mathrm{Beta} (\theta + 1, T).
\end{align}

\subsection{Results}

In order to examine how well our HDP estimation approach performed in comparison with existing methods \cite{etienne07,Etienne09a,Etienne09b}, we used a combination of simulated data and real data that had been analysed before. Firstly, we generated 1,000 simulated data sets of three local samples with 1,000 individuals each for the eight parameter combinations given in Table~\ref{Synthetic_data_sets1}. Note that the migration probability is simply $m_i = I_i/(I_i + J_i - 1)$. These data sets were generated using the PARI/GP code provided in \cite{etienne07}, which is an urn algorithm based on coalescence theory. We then estimated the parameters using the Gibbs sampling approach based on the HDP approximation and the approximate two stage approach of \cite{Etienne09a}. Tables~\ref{HDP_Approx} and~\ref{Etienne_Approx} gives the means, coefficients of variation and mean absolute deviations from the true values of our approach and Etienne's two stage approximate method, respectively, across the 1,000 data sets for each parameter combination.

\begin{table}[!p]
\begin{center}
\begin{tabular}{|c|c|c|c|c|c|c|c|c|}
\hline
Data set & $J_i$ & $\theta$ & $I_1$ & $I_2$ & $I_3$ & $m_1$ & $m_2$ & $m_3$ \\
\hline
1 & 1000 & 5 & 111	 & 249.75	& 666 & 0.1 & 0.2 & 0.4 \\
\hline
2 & 1000	 & 50 & 111 & 249.75 & 666 & 0.1 & 0.2 & 0.4 \\
\hline
3 & 1000	& 500	& 111	& 249.75	& 666 & 0.1 & 0.2 & 0.4  \\
\hline
4 & 1000	 & 5	& 10.0909	 & 52.5789 & 333 & 0.01 & 0.05 & 0.25 \\
\hline
5 & 1000	 & 50 	& 10.0909	 & 52.5789 & 333 & 0.01 & 0.05  & 0.25 \\
\hline
6 & 1000	 & 500 	& 10.0909	 & 52.5789 & 333 & 0.01 & 0.05  & 0.25    \\
\hline
7 &	1000 &	5 &	1 &	2.002 &	4.012 & 0.001 & 0.002  & 0.004   \\
\hline
8 &	1000 &	50 &	1 &	2.002 &	4.012 & 0.001 & 0.002  & 0.004   \\
\hline
\end{tabular}
\end{center}
\caption{The parameter values chosen for the synthetic neutral model data sets that composed our simulation study.}
\label{Synthetic_data_sets1}
\end{table}

\begin{table}[!p]
\begin{center}
\tabcolsep=0.11cm
\footnotesize
\begin{tabular}{|c|c|c|c|c|c|c|c|c|c|c|c|c|}
\hline
Data set & $\hat{\theta}$ & CV & MAD & $\hat{m_1}$ & CV & MAD & $\hat{m_2}$ & CV & MAD & $\hat{m_3}$ & CV & MAD \\
\hline
1 & 5.4092 & 0.20	& 0.8950 & 0.0934 & 0.29 &	0.0232 & 0.1508 & 0.23 &	0.0522 & 0.2002 & 0.19 &	0.1998 \\
\hline
2 & 51.5476 & 0.09	& 3.9993 & 0.0990 & 0.14 &	0.0114 & 0.1923 & 0.15 &	0.0242 & 0.3262 & 0.12 &	0.0749 \\
\hline
3 & 498.8622 & 0.07	& 25.8993 & 0.0999 & 0.08 & 0.0067 & 0.1982 & 0.07 &	0.0119 & 0.3836 & 0.07 &	0.0252 \\
\hline
4 & 5.4477 & 0.22	& 1.0088 & 0.0110 & 0.42 &	0.0032 & 0.0526 & 0.36 &	0.0144 & 0.1417 & 0.26 &	0.1083 \\
\hline
5 & 51.7504 & 0.12	& 4.8836 & 0.0101 & 0.21 &	0.0017 & 0.0504 & 0.17 &	0.0065 & 0.2211 & 0.16 &	0.0387 \\
\hline
6 & 488.8805 & 0.10	& 40.7537 & 0.0100 & 0.17 & 0.0014 & 0.0503 & 0.10 &	0.0040 & 0.2495 & 0.09 &	0.0171 \\
\hline
7 & 5.3388 & 0.46	& 1.8189 & 0.0014 & 0.96 &	0.0007 & 0.0030 & 0.98 &	0.0015 & 0.0066 & 0.95 &	0.0035 \\
\hline
8 & 55.0994 & 0.43	& 17.2483 & 0.0010 & 0.44 & 0.0004 & 0.0022 & 0.34 &	0.0006 & 0.0043 & 0.29 &	0.0009 \\
\hline
\end{tabular}
\end{center}
\caption{Estimates of $\theta$ and $m_i$ from the various scenarios of simulated data sets of Table~\ref{Synthetic_data_sets1} using the hierarchical Dirichlet process approximation. The values reported are the means, coefficients of variation and mean absolute deviations from the true value of the parameter estimates over 1000 such data sets.}
\label{HDP_Approx}
\end{table}

\begin{table}[!p]
\begin{center}
\tabcolsep=0.11cm
\footnotesize
\begin{tabular}{|c|c|c|c|c|c|c|c|c|c|c|c|c|}
\hline
Data set & $\hat{\theta}$ & CV & MAD & $\hat{m_1}$ & CV & MAD & $\hat{m_2}$ & CV & MAD & $\hat{m_3}$ & CV & MAD \\
\hline
1 &  5.9130 & 0.40 & 1.9880 & 0.1899 & 1.45 & 0.1621 & 0.2763 & 1.14 & 0.2300 & 0.4057 & 1.10 & 0.3260 \\
\hline
2 & 51.9033 & 0.20	& 8.2626 & 0.1071 & 0.44 &	0.0274 & 0.2239 & 0.56 & 0.0776 & 0.4231 & 0.48 & 0.1556 \\
\hline
3 & 507.2382 & 0.12 	& 50.4488 & 0.1006 & 0.09 & 0.0070 & 0.2010 & 0.09 & 0.0138 & 0.4032 & 0.12 & 0.0356 \\
\hline
4 & 6.0710 & 0.45	& 2.1911  & 0.0410 & 3.62 &	0.0356 & 0.1177 & 1.88 & 0.1042 & 0.3086 & 1.11 & 0.2666 \\
\hline
5 & 54.2026 & 0.29	& 12.6540  & 0.0102 & 0.55 & 0.0020 & 0.0580 & 0.83 &	0.0190 & 0.2897 & 0.72 & 0.1440 \\
\hline
6 & 578.4131 & 0.36	& 166.5742  & 0.0100 & 0.18 & 0.0014 & 0.0503 & 0.13 & 0.0048 & 0.2601 & 0.34 & 0.0503 \\
\hline
7 & 9.9517 & 1.41	& 6.5506 & 0.0164 & 7.03 &	0.0158 & 0.0348 & 4.69 & 0.0338 & 0.0473 & 3.88 & 0.0450 \\
\hline
8 & 860.1590  & 7.00	& 824.9333 & 0.0011 & 1.61 & 0.0004 & 0.0022 & 0.73 &	0.0007 & 0.0075 & 6.32 & 0.0045 \\
\hline
\end{tabular}
\end{center}
\caption{Estimates of $\theta$ and $m_i$ from the various scenarios of simulated data sets of Table~\ref{Synthetic_data_sets1} using Etienne's approximate method. The values reported are the means, coefficients of variation and mean absolute deviations from the true value of the parameter estimates over 1000 such data sets.}
\label{Etienne_Approx}
\end{table}

For all parameter combinations considered the HDP approximation outperforms Etienne's approximation as an estimator of $\theta$, as in each case the overall means are closer to the true values and the coefficients of variations and mean absolute deviations from the true values are considerably smaller. The HDP approximation provides a less biased and more reliable estimator of $\theta$ than Etienne's approximation.

A similar pattern is observed with the estimates of the immigration probabilities $m_i$, as for the parameter combinations considered our approach gives lower coefficients of variation and mean absolute deviations from the true value than Etienne's approximate method. Both approximations break down when the immigration rate $I$ is significantly larger than the fundamental biodiversity parameter $\theta$ (for example, see the estimates of $m_3$ for synthetic data sets 1-5 in Tables~\ref{HDP_Approx} and~\ref{Etienne_Approx}), but in different ways. Our method underestimates the immigration probability $m$ in such cases, but the standard deviation around that estimate remains low, and thus our estimator for $m$ is biased when $I > \theta$, but as this bias is consistent it would be possible to correct for it. On the other hand, Etienne's approximate approach gives an overall mean over the 1,000 simulated data sets that is much closer to the true value in such a case than our method does. However, the variability around Etienne's approximate estimate of $m$ is much higher because the algorithm often converges to an immigration probability of 1, even when the true value is much lower. It is also worth noting that Etienne's approximate method also breaks down badly for data sets 7 and 8 where the immigration probabilities are very low, whereas the HDP approximation copes much better in such scenarios. Thus, we conclude that the HDP approximation is a better estimator of the neutral model's parameters than Etienne's approximation unless $I >> \theta$ and the immigration probabilities are close to 1. 

In Table~\ref{Computation_time}, we present the average times in seconds of Etienne's approximate method using the code given in \cite{Etienne09a} and PARI/GP's default settings, and our Gibbs sampling approach coded in C++ when it was run for 50,000 iterations with half of these being conservatively discarded as burn-in. Under these settings, for all but one of the simulated data set scenarios of Table~\ref{Synthetic_data_sets1}, Etienne's approximate method is two to three times faster than our approach. However, we are being very conservatie with sample number and equivalent results could be achieved with as little at 10,000 iterations when the two methods would be of comparable speed.

\begin{table}[!p]
\begin{center}
\begin{tabular}{|c|c|c|c|}
\hline
Data set & Etienne's approximation & HDP approximation \\
\hline
1 & 13.8583 & 40.6223  \\
\hline
2 & 21.5615 & 41.1254  \\
\hline
3 & 208.6595 & 41.5881 \\
\hline
4 & 14.9588 & 41.8532  \\
\hline
5 & 14.9767 & 40.6765  \\
\hline
6 & 27.3442 & 42.4084  \\
\hline
7 & 20.0091 & 56.1613  \\
\hline
8 & 17.8649 & 57.5658 \\
\hline
\end{tabular}
\end{center}
\caption{Average time in seconds that Etienne's approximate method and the HDP approximation took to run on the various scenarios of simulated data sets of Table~\ref{Synthetic_data_sets1}. Note that the HDP approximation was run for 50,000 iterations and half of these were conservatively discarded as burn-in.}
\label{Computation_time}
\end{table}

We were unable to replicate these results using Etienne's `exact' maximum likelihood method, so instead we quote those that he gave in a similar simulation study \cite{Etienne09b} in Table~\ref{Etienne_Exact}. We see that Etienne's `exact' method slightly outperforms the HDP approximation as an estimator of $\theta$, as although the coefficients of variation are broadly similar, the overall means are generally closer to their true values and thus Etienne's `exact' method is less biased for this parameter. Regarding the estimation of immigration probabilities, the results are comparable when $\theta <= I$. When $\theta > I$, there is a tendency for Etienne's `exact' method to overestimate the immigration probability, but not as badly as the HDP approximation underestimates it. The advantage of the HDP approximation is that our code is easier to implement than Etienne's `exact' method's PARI/GP algorithm, it is much faster, and our approach can handle the large data sets often encountered in microbiomics.

\begin{table}[!p]
\begin{center}
\begin{tabular}{|c|c|c|c|c|c|c|c|c|c|c|c|c|}
\hline
Data set & $\hat{\theta}$ & CV  & $\hat{m_1}$ & CV  & $\hat{m_2}$ & CV & $\hat{m_3}$ & CV  \\
\hline
1 & 4.9689 & 0.21 & 0.1119 & 0.44 & 0.2353 & 0.49	& 0.4727 & 0.50 \\
\hline
2 & 49.9838 & 0.10	& 0.1022 & 0.16 & 0.2041	& 0.16 & 0.4105 & 0.18 \\ 
\hline
3 & 501.5142 & 0.07 & 0.1005 & 0.08 & 0.2009 & 0.08	& 0.4007 & 0.08 \\ 
\hline
4 & 4.8982 & 0.25 & 0.0108 & 0.43 & 0.0572 & 0.46	& 0.3658 & 0.70 \\ 
\hline
5 & 49.9892 & 0.12 & 0.0103 & 0.21	& 0.0513 & 0.16 & 0.2643	& 0.25 \\ 
\hline
6 &  504.0792 & 0.11 & 0.0101	& 0.17 & 0.0504 & 0.11 & 0.2521 & 0.09 \\ 
\hline
7 & 5.0388 & 0.45 & 0.0012 & 0.67	& 0.0027 & 1.27 & 0.0066 & 4.85 \\
\hline
8 & 56.0378 & 0.55	& 0.0010 & 0.42 & 0.0020	& 0.35 & 0.0042 & 0.30 \\
\hline
\end{tabular}
\end{center}
\caption{Estimates of $\theta$ and $m_i$ from the various scenarios of simulated data sets of Table~\ref{Synthetic_data_sets1} using Etienne's `exact' maximum likelihood method. The values reported are the means and coefficients of variation over 1000 such data sets, and were obtained from \cite{Etienne09b}.}
\label{Etienne_Exact}
\end{table}

As an example of how the methods compare on real data, we reanalysed the tropical tree data set used as an example in \cite{etienne07,Etienne09a,Etienne09b}. The data consists of three forest plots in Panama called Barro Colorado Island (50 ha), Cocoli (4 ha) and Sherman (5.96 ha), which lie along a precipitation gradient \cite{Condit}. Table~\ref{Panama_tree_results} shows the results of the parameter estimation for Etienne's three methods and our HDP approach. We see that in this case the results from the HDP approximation closely match Etienne's `exact' method, while his approximate method overestimates $\theta$ and underestimates the immigration rates. The matching results of our approach and Etienne's `exact' method is unsurprising as in this case $\theta >> I_i$. 

\begin{table}[!p]
\centering
\begin{tabular}{|c|c|c|c|c|}
\hline
Method & $\theta$ & $I_{BCI}$ & $I_C$ & $I_S$ \\
\hline
Etienne fixed $I$ & 259 & 44.2 & 44.2 & 44.2 \\ 
\hline
Etienne approx & 342 & 53.7 & 30.8 & 33.9 \\ 
\hline
Etienne `exact' & $235 \pm 23$ & $65.3 \pm 5.9$ & $31.5 \pm 3.9$ & $35.7 \pm 3.9$ \\
\hline
HDP approx & $231 \pm 22$ & $65.5 \pm 5.9$ & $31.6 \pm 3.8$ & $35.8 \pm 3.9$ \\
\hline
\end{tabular}
\caption{Neutral parameter estimates for samples from three local tree communities (Sherman, BCI and Cocoli) in the Panama Canal Zone using Etienne's approaches and the hierarchical Dirichlet process approximation. Standard errors are given for the methods where they are available.}
\label{Panama_tree_results}
\end{table}

\bibliography{UNTB-HDP}

\end{document}

%% file: GlobalDefs.tex
\usepackage{amsmath,amssymb,amsthm,bbm,color,enumerate,mathrsfs,setspace}

\newcommand{\bbN}{\mathbb{N}}
\newcommand{\bbE}{\mathbb{E}}
\newcommand{\bbP}{\mathbb{P}}

\newcommand{\bbD}{\mathbb{D}}

\newcommand{\abs}[1]{\left| #1 \right|}
\newcommand{\norm}[1]{\left\| #1 \right\|}
\newcommand{\BigO}[1]{\mathcal{O}{\textstyle\left( #1\right)}}

\newcommand{\Th}[1]{$#1$\textsuperscript{th}}

\newcommand{\defn}{:=}

\newcommand{\bx}{\mathbf{x}}
\newcommand{\by}{\mathbf{y}}

\newcommand{\bG}{\mathbf{G}}
\newcommand{\bH}{\mathbf{H}}

\newcommand{\bX}{\mathbf{X}}

\newcommand{\bmu}{\boldsymbol \mu}

\newcommand{\cT}{\mathcal{T}}

\newcommand{\cG}{\mathcal{G}}

\newtheorem{lem}{Lemma}
\newtheorem{prop}{Proposition}
\newtheorem{cor}{Corollary}
\theoremstyle{remark} 
	\newtheorem{rem}{Remark}
	
\theoremstyle{definition} 
	 
	\newtheorem{exmp}{Example}

\newcommand{\ie}{\textit{i.e.,}\,}
\newcommand{\eg}{\textit{e.g.,}\,}

\newcommand{\nb}{\textit{n.b.,}\, }